\documentclass[preprint]{elsarticle}

\usepackage[utf8]{inputenc}
\usepackage{xspace}
\usepackage[english]{babel}
\usepackage{amsfonts}
\usepackage{amsmath}
\usepackage[caption=false]{subfig}
\usepackage{enumitem}
\usepackage{textcomp}
\usepackage[normalem]{ulem}
\setlist{topsep=2pt,parsep=0pt,partopsep=0pt,itemsep=2pt}
\setlist[1]{labelindent=\parindent}
\setlist[2]{topsep=0pt}

\journal{Theoretical Computer Science}

\newdefinition{definition}{Definition}[section]
\newdefinition{remark}{\normalfont \it Remark}[section]
\newdefinition{example}{\normalfont \it Example}[section]
\newtheorem{theorem}{Theorem}[section]
\newtheorem{proposition}[theorem]{Proposition}
\newtheorem{lemma}[theorem]{Lemma}
\newtheorem{corollary}[theorem]{Corollary}

\newproof{proof}{\textit{Proof}}
\newproof{sproof}{\textit{Sketch of proof}}

\newcommand{\ie}{i.e.,\xspace}
\newcommand{\eg}{e.g.,\xspace}
\newcommand{\etc}{etc.\xspace}

\newcommand{\Z}{\mathbb{Z}}
\newcommand{\N}{\mathbb{N}}
\renewcommand{\P}{\mathcal{P}}
\newcommand{\T}{\tau}
\newcommand{\R}{\mathcal{R}}
\newcommand{\rib}{\rho}

\newcommand{\Ribs}{\mathfrak{R}}
\newcommand{\Zips}{\mathfrak{Z}}
\newcommand{\Tils}{\mathfrak{T}}
\newcommand{\NVN}{N_{vN}}
\DeclareMathOperator{\dom}{dom}
\DeclareMathOperator{\range}{range}

\newcommand{\set}[1]{\ensuremath{\left\{#1\right\}}}
\newcommand{\card}[1]{\ensuremath{\left|#1\right|}}
\newcommand{\abs}[1]{\ensuremath{\left|#1\right|}}
\newcommand{\st}{\mid}
\newcommand{\scale}[1]{\left|#1\right|}

\newcommand{\ignore}[1]{}

\begin{document}

\begin{frontmatter}

\title{Polyominoes Simulating Arbitrary-Neighborhood\\Zippers and Tilings\tnoteref{FNANO,sponsors}}
\tnotetext[FNANO]{Some results in Section~\ref{sec:ziprib} of this paper were presented at FNANO 2009 conference~\cite{KM09}.}
\tnotetext[sponsors]{This research was supported by a Natural Sciences and Engineering Research Council of Canada discovery grant, and by a Canada Research Chair award to L.K.}
\author[UWO]{Lila Kari\corref{cor}}
\ead{lila@csd.uwo.ca}

\author[UWO]{Beno\^it Masson}
\ead{benoit@csd.uwo.ca}

\cortext[cor]{Corresponding author.}
\address[UWO]{University of Western Ontario -- Department of Computer Science\\Middlesex College, London, Ontario, Canada, N6A 5B7}

\begin{abstract}
This paper provides a bridge between the classical tiling theory and the complex neighborhood self-assembling situations that exist in practice.

The neighborhood of a position in the plane is the set of coordinates which are considered adjacent to it. This includes classical neighborhoods of size four, as well as arbitrarily complex neighborhoods. A generalized tile system consists of a set of tiles, a neighborhood, and a relation which dictates which are the ``admissible'' neighboring tiles of a given tile. Thus, in correctly formed assemblies, tiles are assigned positions of the plane in accordance to this relation.

We prove that any validly tiled path defined in a given but arbitrary neighborhood (a zipper) can be simulated by a simple ``ribbon'' of microtiles. A ribbon is a special kind of polyomino, consisting of a non-self-crossing \ignore{rectilinear} sequence of tiles on the plane, in which successive tiles stick along their adjacent edge.\ignore{, and where each tile needs to match glues with only two other tiles: its predecessor and its successor on the path. }\ignore{Our constructions simulate each of the existing tiles by a polyomino of microtiles, whose shape is used to simulate the given tile and the communication of information between itself and its neighbors. The polyominoes can then be catenated together to simulate the entire complex-neighborhood tiled path by a continuous two-tile-neighborhood ribbon.}

Finally, we extend this construction to the case of traditional tilings, proving that we can simulate arbitrary-neighborhood tilings by simple-neighborhood tilings, while preserving some of their essential properties. 
\end{abstract}

%\subjclass{F.1.1 Models of Computation}

\end{frontmatter}

\section{Introduction}

% General intro
Because of the constant miniaturization of components, microscopic elements in the fields of electronics, engineering or even medicine are becoming more and more difficult to construct and to manipulate. A recent approach to work around this problem is \emph{self-assembly}. It consists in ``programming'' the nano-components, so that when starting from an initial pool of selected components, they automatically aggregate to form bigger and bigger structures, until they eventually form a final complex structure.

The first formal models for self-assembly were introduced a decade ago~\cite{W98,WLWS98,A00}. In this framework, self-assembling components were modeled by Wang tiles~\cite{W61}, \ie unit squares that cannot be rotated and that have ``glues'' on each of their four edges. Two tiles stick to each other if they have the same glue on their common edge. By carefully designing the glues, and starting with an initial tile called ``seed'', complex structures can self-assemble.

The use of this simple model as a formalization of the process of self-assembly allowed the application for theoretical studies of dynamical self-assembly of many well-known existing results and techniques concerning ``static'' tilings~\cite{W61} and cellular automata~\cite{K05}, such as the undecidable problem of the tiling of the plane, and the simulation of a Turing machine by a tile system~\cite{B66,R71}.

\smallskip

Most of the theoretical results on self-assembly presume that each tile interacts via glues only with tiles in its so-called \emph{von Neumann neighborhood}, which includes the four tiles situated at the North, South, East, and West of the tile itself~\cite{RW00}%
\footnote{Total tilings can also be seen as bidimensional subshifts of finite type, and studied as such in the field of symbolic dynamics (see~\cite{LM95} for a complete introduction, in dimension 1). In this settings, the ``neighborhood'' is arbitrary, by definition of the subshift. Although some results exist for subshifts in dimensions greater than 1~\cite{S93,W94,H09}, they do not apply to partial tilings.}%
. Only relatively few recent results consider more general cases, such as larger neighborhood~\cite{AKKRS09}, or a three-dimensional neighborhood~\cite{BRS08}. Even the most well-known experimental incarnation of square tiles, the DNA tiles~\cite{WLWS98,WYS98,MLRS00,RPW04,FHPWM07}, deal only with the von Neumann-sized neighborhood, where the DNA single strands located at the corner of each rectangular DNA tile allow its interaction with four neighbors only. Other experimental situations that could be modeled by self-assembly of tiles, such as atomic or molecular interactions, potentially include more complex scenarios where the neighborhood of a tile is both larger and more complex than the von Neumann neighborhood. At the limit, one can consider the case of an arbitrary neighborhood where tiles that are not physically adjacent to the main tile may be its neighbors, while some adjacent tiles may not.

In~\cite{AKKRS09}, it was proved that, for any directed tile system, any von Neumann-neighborhood ``zipper'' (a tiled rectilinear path) can be simulated by a ``ribbon'' constructed with tiles from a new tile system. A ribbon is a non-self-crossing rectilinear succession of tiles in the plane, where each tile is required to have glues matching with two tiles only: its predecessor and its successor on the ribbon-path. The construction that simulated a directed von Neumann-neighborhood tiled path by a ribbon, replaced each of the existing tiles by so-called ``motifs'' which traced the contours of the initial tile but where, in addition, bumps and matching dents along the motif edges simulated both the matching of the glues and the directionality of the path. In other words, geometry of the motifs was used to simulate glue matching. Note that motifs, as well as ribbons, are particular cases of \emph{polyominoes}~\cite{SW05,BN91}, \ie finite and connected sets of DNA tiles. 

\smallskip

This polyomino construction led to a conjecture by Jarkko Kari, claiming that it is possible to ``simulate'' an arbitrary-neighborhood zipper by a simple two-tile-neighborhood ribbon. A first step in this direction was~\cite{CK10,CK09}, wherein it was proved that a complex-neighborhood zipper, defined for example on the Moore neighborhood (von Neumann plus four diagonal neighbors), can be simulated by a ribbon of irregularly shaped tiles, where the shape was used to simulate the neighborhood relationship.

The aim of this paper is to answer the above conjecture positively for the case of arbitrary-neighborhood zippers, thus providing a bridge between the classical work in tiling theory or cellular automata and the realistic complex-neighborhood self-assemblies that exist in practice. We namely prove in Corollary~\ref{cor:main} that for \emph{any} neighborhood, zippers can be simulated by simple polyominoes connected to each other end-to-end to form a ribbon that essentially traces the same path. The main idea used in our simulation is that each existing tile can be replaced by a polyomino, where the shape of the polyomino is used to simulate the communication between a tile and its adjacent or distant neighbors. We also show that, by the design of the shapes of the polyominoes, we can transmit information at a distance, sometimes across other information pathways, without violating the non-self-crossing feature of the ribbon. Such situations where information pathways cross are inherent in, for example, Moore neighborhoods where, \eg the communication channel between a tile and its Northeast neighbor ``crosses'' the communication channel between its North neighbor and its East neighbor.

We also explain how the simulation of zippers by ribbons can be modified to simulate arbitrary-neighborhood tilings by von Neumann-neighborhood tilings. The idea is to modify the polyominoes, adding new constraints so that the two-tile neighborhood can be replaced by a von Neumann neighborhood (Corollary~\ref{cor:mainT}). The main significance of this simulation, as opposed to more intuitive ones where some ``supertiles'' are used to transfer information, is that it applies to all tilings, even when they are partial. Besides, some essential properties of the initial tiling are preserved by the simulation. We prove that this is the case for partial and periodic tilings, and in a more restricted setting, for convex tilings.

\smallskip

The paper is organized as follows. In Sect.~\ref{sec:def}, we recall basic definitions and give a formal definition of a simulation. Then, Sect.~\ref{sec:ziprib} describes our construction in the case of zippers, starting with the general idea of simulating a zipper by a ribbon, by sketching the proof from~\cite{AKKRS09}, in the simple case of a von Neumann neighborhood. We also highlight the technical difficulties related to crossing of information pathways, that prevent this technique from being transferable without modifications to the case of the Moore neighborhood~\cite{CK10}. Then, we prove our result for the case of arbitrary linear neighborhoods. This construction is eventually generalized to arbitrary neighborhoods to prove the main result of this section. Finally, in Sect.~\ref{sec:tilings}, we adapt the construction to the simulation of regular tilings, and we study how partial, periodic, line or column convex properties are preserved.

%%%%%%%%%%%%%%%%%%%%%%%%%%%% DEFINITIONS %%%%%%%%%%%%%%%%%%%%%%%%%%%%

\section{Definitions}\label{sec:def}

First, we give some basic definitions on tilings, and then we introduce more technical definitions to describe the principles of simulation by polyominoes. We finally illustrate this on a concrete example, the simulation of total tilings.

\subsection{Tilings, Ribbons, and Zippers}

  Historically, Wang tiles~\cite{W61} were defined as oriented unit squares, on the border of which were 4 \emph{glues} (colors) used to stick them to neighboring tiles, provided the glues matched. We generalize this notion to arbitrary neighborhoods~\cite{K05,AKKRS09}, \ie neighborhoods which can contain other tiles than the North, South, West and East tiles. In this paper, a \emph{neighborhood} $N \subset \Z^2$ is the set of relative coordinates of the neighboring tiles, such that
\begin{itemize}
  \item $N$ is finite;
  \item $(0, 0) \not\in N$ (a tile can not be one of its neighbors);
  \item $(i, j) \in N \Rightarrow (-i, -j) \in N$ (if a tile $t'$ is a neighbor of a tile $t$, then $t$ has to be a neighbor of $t'$).
\end{itemize}
For example, the usual 4-tile neighborhood (called \emph{von Neumann} neighborhood~\cite{K05}) is $\NVN = \set{(0,1), (1,0), (0,-1), (-1,0)}$; the 8-tile \emph{Moore} neighborhood is $N_M = \set{(0,1), (1,1), (1,0), (1,-1), (0,-1), (-1,-1), (-1,0), (-1,1)}$.

For the following definitions, a neighborhood $N$ and a finite set of glues $X$ are fixed. A \emph{tile} is a tuple $t = \left(t_{i,j}\right)_{(i,j) \in N}$ where each $t_{i,j} \in X$. Intuitively, $t_{i,j}$ is the glue that will be used to match the neighbor located at position $(i,j)$ relative to the tile $t$. A \emph{tile system} $T$ is a finite set of tiles, used to build larger structures. A tile $t \in T$ \emph{sticks} at position $(i,j)$ to a tile $t' \in T$ if the corresponding glues match, \ie $t_{i,j} = t'_{-i,-j}$.

A \emph{$(T,N)$-tiling} using tile system $T$ in neighborhood $N$ (or \emph{$N$-neighborhood tiling}, or simply \emph{tiling} when the tile system and its neighborhood are known without ambiguity) is a mapping $\T : D \to T$, where $D \subset \Z^2$ is a subset of the plane, which associates every position $(x,y) \in D$ with a tile, such that all tiles stick to their neighbors, \ie for all $(x,y) \in D$, for all $(i,j) \in N$ such that $(x+i,y+j) \in D$, $\T(x,y)_{i,j} = \T(x+i,y+j)_{-i,-j}$. Note that tilings are often referred to as ``valid'' tilings in the literature. The set of all $(T,N)$-tilings is denoted $\Tils_{T,N}$.

When $D = \Z^2$, the tiling is said to be \emph{total}, otherwise it is \emph{partial}. In addition, if $D$ is finite then the partial tiling is also \emph{finite}. A tiling $\T: D \to T$ is \emph{connected} if its domain $D$ is $4$-connected. As suggested in~\cite{SW05}, we call \emph{polyomino} a finite and connected tiling\footnote{The usual definition of a polyomino (see for example~\cite{BN91}) refers to a domain $D \subset \Z^2$, while here we add to this notion a mapping to tiles. Besides, polyominoes are often defined as simply connected sets of squares, \ie without holes, which is not required in this paper.}.

According to the usual definition, a total tiling $\T$ is \emph{periodic} if it admits a horizontal and a vertical \emph{period} $p_h, p_v \in \N_+$, \ie for all $x, y \in \Z$, $\T(x,y) = \T(x+p_h,y) = \T(x,y+p_v)$. We extend this definition to partial tilings as follows: a tiling $\T: D \to T$ is periodic if it admits a horizontal and a vertical period $p_h, p_v \in \Z$ such that for all $(x, y) \in D$ and for all $\alpha, \beta \in \Z$, $(x+\alpha p_h,y+\beta p_v) \in D$ implies $\T(x,y) = \T(x+\alpha p_h,y+\beta p_v)$. Note that with this extended definition, all finite tilings are periodic (the periods should be ``larger'' than the tiling itself). 

A tiling $\T: D \to T$ is \emph{line convex} [resp., \emph{column convex}] if $(x_1, y), (x_2,y) \in D$ and $x_1 < x_2$ [resp., $(x, y_1), (x,y_2) \in D$ and $y_1 < y_2$] imply that for all $x_1 < x < x_2$ [resp., for all $y_1 < y < y_2$], $(x,y) \in D$. A tiling is \emph{convex} if it is both line and column convex.

\smallskip

Two positions of the plane $u,v \in \Z^2$ (or, by extension, tiles of a tiling) are said to be \emph{adjacent} when they are neighbors in the von Neumann sense, \ie $v-u \in \NVN$. A \emph{path} is a sequence of adjacent positions of the plane. Formally, a path is a function $P : I \to \Z^2$, where $I \subset \Z$ is a set of consecutive integers, such that for all $i,i+1 \in I$, $P(i)$ and $P(i+1)$ are adjacent. For a given tile system $T$, a \emph{$T$-tiled path} using tile system $T$ is a sequence of adjacent tiles from $T$, \ie a pair $(P,r)$ where $P$ is a path and $r : \range(P) \to T$ a mapping from positions to tiles. We say that a $T$-tiled path is \emph{finite} when $\dom(P)$ is finite, and we may also call a finite $T$-tiled path a \emph{polyomino} since it is connected (by definition of the path).

For a tile system $T$ defined in von Neumann neighborhood $\NVN$, a $T$-tiled path $(P,r)$ is a \emph{$T$-ribbon} using tile system $T$ (simply called \emph{ribbon} when $T$ is known without ambiguity) if $P$ is injective (non-self-crossing) and for all $i, i+1 \in \dom(P)$, $r(P(i))_v = r(P(i+1))_{-v}$, with $v = P(i+1)-P(i) \in \NVN$. The glue $r(P(i))_v$ is called the \emph{output} glue of $r(P(i))$, while $r(P(i+1))_{-v}$ is the \emph{input} glue of $r(P(i+1))$. Informally, a ribbon is a sequence of adjacent tiles which stick to their predecessor and successor only (see Fig.~\ref{fig:ribzip}\subref{fig:ribbon}). Consequently, $r$ is not necessarily a tiling.

A $T$-tiled path $(P,r)$ is a \emph{$(T,N)$-zipper} using tile system $T$ in neighborhood $N$ (or \emph{$N$-neighborhood zipper}, or \emph{zipper} when $T$ and $N$ are known) if $P$ is injective and $r$ is a $(T,N)$-tiling. A zipper can be seen as a tiling with an additional notion of unique input and output for every tile (except the first which has only an output and the last which has only an input), each input being connected to the output of an adjacent tile (see Fig.~\ref{fig:ribzip}\subref{fig:zipper}). Note that zippers can be defined in arbitrary neighborhoods, since it is not required that adjacent glues match. The set of $T$-ribbons is denoted $\Ribs_T$, the set of $(T,N)$-zippers $\Zips_{T,N}$.

\begin{figure}[!ht]
  \begin{center}
  \hfill
  \subfloat[A ribbon.]{\label{fig:ribbon}
    \includegraphics{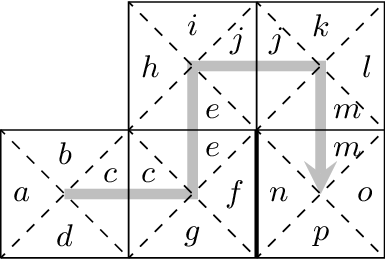}
}
  \hfill\hfill
  \subfloat[A zipper.]{\label{fig:zipper}
    \includegraphics{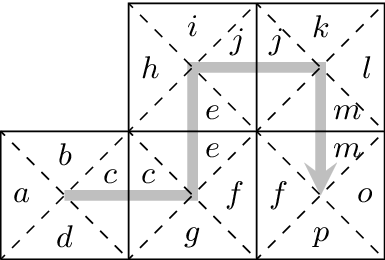}
}
  \hfill\null
  \caption{A ribbon, and a von Neumann-neighborhood zipper. The underlying path is drawn in gray. The only difference occurs at the bottom-right where glues $f$ and $n$ are not required to match in the ribbon, while in the zipper they do.}
  \label{fig:ribzip}
  \end{center}
\end{figure}

%%% Poly-trucs and simulations

\subsection{Poly-Tilings and Simulations}

In this section, we give the main idea of the simulation of a tiling by another, as well as for the simulation of a zipper by a ribbon. Informally, a simulation is a function which associates each tiling [resp., zipper] using the ``initial'' tile system in the ``initial'' neighborhood with a tiling [resp., ribbon] using a ``new'' tile system in a ``new'' neighborhood. The simulation must also allow to recover the initial object without ambiguity. To build this new object, we first associate every initial tile with unique polyominoes. Then, we catenate these polyominoes to form the new, bigger, object, called \emph{poly-tiling} or \emph{poly-ribbon}.

In the sequel, for a set $D \in \Z^2$ and a vector $u$, we denote $D + u = \set{(x,y) + u \st (x,y) \in D}$ the set $D$ translated by $u$.

\begin{definition}\label{def:macrotiling}
  A \emph{poly-tiling} $\T$ of \emph{scale} $s$\ignore{ (also denoted as $\scale{\T}$)}, using tile system $T$ in neighborhood $N$, is a mapping $\T : D \to \Tils_{T,N}$, with $D \subset \Z^2$, such that
\begin{enumerate}[label=(\emph{\roman*})]
\item\label{item:macro1} for all $u \in D$, $\T(u)$ is a finite $(T,N)$-tiling;
\item\label{item:macro2} for all $u,v \in D$, $u \ne v$, $[\dom(\T(u)) + su] \cap [\dom(\T(v)) + sv] = \emptyset$;
\item\label{item:macro3} for all $u, u' \in D$, $v \in \dom(\T(u))$ and $v' \in \dom(\T(u'))$, let $w = (su'+v')-(su+v)$; then $w \in N$ implies $\T(u)(v)_{w} = \T(u')(v')_{-w}$.
\end{enumerate}
\end{definition}

Let us discuss the items in this definition from an informal point of view, as shown in Fig.~\ref{fig:poly-tiling}. We refer only to the particular case in which all $\T(u)$ are connected and thus polyominoes, since this is what we will construct. In this case, the definition implies that a poly-tiling of scale $s$ is
\begin{itemize}
\item[\ref{item:macro1}] a juxtaposition of polyominoes on a grid of size $s$;
\item[\ref{item:macro2}] such that all these polyominoes do not overlap once shifted to their actual position on the grid, which is achieved by shifting $\T(x,y)$ by $sx$ units to the right and $sy$ units upwards;
\item[\ref{item:macro3}] neighboring polyominoes stick, \ie if two tiles of two different polyominoes become neighbors after being shifted, they have to stick.
\end{itemize}

\begin{figure}[!ht]
\begin{center}
\includegraphics{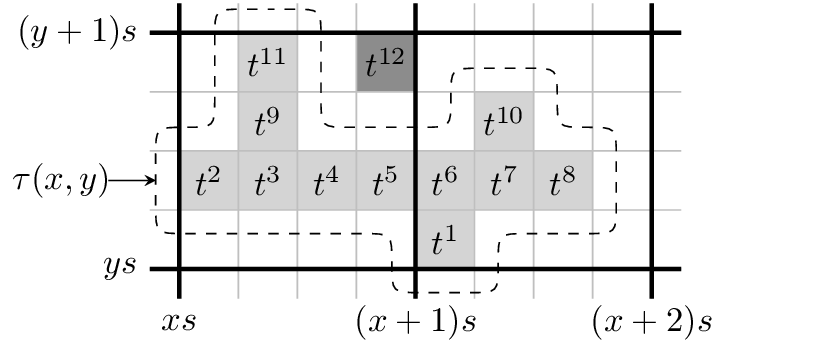}
\caption{Poly-tiling $\T$ of scale $s=4$ seen as a tiling. The dashed line surrounds the tiling $\T(x,y)$. Note that some tiles from $\T(x,y)$ may be located outside of the ``box'' of size $s$ situated at $(sx,sy)$, and conversely that some  tiles of this box (here, $t^{12}$) may not belong to $\T(x,y)$.}
\label{fig:poly-tiling}
\end{center}
\end{figure}

These characteristics allow to consider them as usual ``valid'' tilings without overlaps, as explained in the following remark.

\begin{remark}\label{rem:macrotiling}
  A poly-tiling $\T$ of scale $s$ can be easily transformed into a ``regular'' tiling $\T'$, as depicted in Fig.~\ref{fig:poly-tiling}, according to the following formula:
\[
  \T'(i,j) = (\T(x,y))(i', j')\enspace,
\]
for any $x,y, i, j, i', j' \in \Z$ which verify $i = sx+i'$, $j = sy+j'$ and $(i',j') \in \dom(\T(x,y))$. Because of condition~\ref{item:macro2} in the definition of poly-tilings, given any pair $(i,j)$, if $x, y, i', j'$ exist then they are unique. Because of~\ref{item:macro1} and~\ref{item:macro3}, $\T'$ is indeed a tiling since all glues match.
\end{remark}

We can slightly modify this notion to define poly-ribbons. Basically, we replace the property of sticking at the macro-level~\ref{item:macro3} by a notion of macro-path.

\begin{definition}\label{def:macrorib}
  A \emph{poly-ribbon}\ignore{ $\rib$} of \emph{scale} $s$\ignore{ (also denoted as $\scale{\rib}$)}, using tile system $T$ (in neighborhood $\NVN$), is a pair $(P,r)$ where $P$ is an injective path and $r : \range(P) \to \Ribs_{T}$ is such that
\begin{enumerate}[label=(\emph{\roman*})]
\item\label{item:macrorib1} for all $u \in \range(P)$, $r(u)$ is a finite $T$-ribbon;
\item\label{item:macrorib2} for all $u,v \in \range(P)$, $u \ne v$, $[\dom(r(u)) + su] \cap [\dom(r(v)) + sv] = \emptyset$ (where we define $\dom(r(u)) = \dom(r')$, if $r(u) = (P',r')$);
\item\label{item:macrorib3} for all $i, i+1 \in \dom(P)$, let $r(P(i)) = (P_1,r_1)$ and $r(P(i+1)) = (P_2,r_2)$ be two consecutive ribbons. Let $u_1 = P_1(\max(\dom(P_1)))$ be the last position of $P_1$, $u_2 = P_2(\min(\dom(P_2)))$ be the first position of $P_2$, and $v = [u_2+sP(i+1)] - [u_1+sP(i)]$ be their relative position. Then, we must have $v \in \NVN$ and $r_1(u_1)_v = r_2(u_2)_{-v}$.
\end{enumerate}
\end{definition}

The new requirement~\ref{item:macrorib3} enforces that the ribbon follows a path at the macro-level, ``jumping'' from one ribbon to the other. Therefore, in a similar way as what was done in Remark~\ref{rem:macrotiling}, we can see a poly-ribbon as a ribbon. This explains why in our constructions we build poly-tilings [resp., poly-ribbons] for simplicity, they can in turn be considered as particular tilings [resp., ribbons] from the set $\Tils_{T,N}$ [resp., $\Ribs_T$].

\smallskip

Let us now formalize the intuition that a simulation is a function $\psi$ which associates tilings with unique poly-tilings [resp., zippers with unique poly-ribbons], by means of a function $\varphi$ transforming every tile into a set of unique polyominoes. These polyominoes will be aligned on a grid of size $s$, each of them replacing a tile from the initial object at the same position, to form a poly-tiling [resp., poly-ribbon] corresponding to the initial tiling [resp., zipper].

In the following, for a set $S$, we define $\P(S) = 2^S$ as the set of subsets of $S$. Let $\varphi : T \to \P(\Tils_{T',N'})$ be a function which associates a tile $t$ with a set of finite $(T',N')$-tilings of scale $s$, such that for all $t \in T$, for all $\T \in \varphi(t)$, $(0,0) \in \dom(\T)$ (the tile $\T(0,0)$ is called the \emph{reference} tile), and such that if $t \ne t'$, for all $\T \in \varphi(t)$, $\T' \in \varphi(t')$, $u \in \dom(\T)$, and $u' \in \dom(\T')$, then $\T(u) \ne \T'(u')$ (all tiles from all tilings in $\varphi(t)$ are specific to tile $t$). In our constructions, the finite tilings $\varphi(x,y)$ will be connected \emph{polyominoes}. Let us introduce one more mapping $\chi_{\varphi} : \Tils_{T',N'} \to \Tils_{T',N'}$, which associates any $(T',N')$-tiling $\T$ with a poly-tiling $\chi_{\varphi}(\T)$, by removing incomplete or misaligned polyominoes as follows.
\begin{enumerate}
\item\label{it:simu_inv_1} First, discard all the tiles from $\T$ which do not belong to a complete polyomino $\varphi(t)$, for some $t \in T$. We obtain $\T_1$ which is a restriction of $\T$.
\item\label{it:simu_inv_2} Then, discard all the polyominoes in $\T_1$ which do not have their reference tile positioned at $(sx,sy)$ for some $x, y \in \Z$ (this is to avoid misalignment in non-connected tilings), to obtain $\chi_{\varphi}(\T)$.
\end{enumerate}
Note that $\chi_{\varphi}(\T)$ is the empty tiling for infinitely many tilings (for example, all the ones for which the reference tiles are badly positioned).

\begin{definition}\label{def:tilingsimu}
  Given such functions $\varphi$ and $\chi_{\varphi}$, a \emph{tiling-simulation} of a tile system $T$ by a tile system $T'$ is a mapping $\psi : \Tils_{T,N} \to \P(\chi_{\varphi}(\Tils_{T',N'}))$, such that $\psi$ associates any $(T,N)$-tiling $\T : D \to T$ with a non-empty subset of poly-tilings (hence tilings) from $\{\T': D \to \Tils_{T',N'} \st \T'(x,y) \in \varphi(\T(x,y))\}$; with the additional constraint that for all $(T', N')$-tilings $\T'$, there exists some unique valid $(T,N)$-tiling $\T$ such that $\chi_{\varphi}(\T') \in \psi(\T)$.
\end{definition}

Note that the constraint on $\varphi$ saying that tiles are specific implies the injectivity of $\varphi$ and of $\psi$. In fact, we even have the following stronger result: $\psi(\T) \cap \psi(\T') = \emptyset$ for two different $(T,N)$-tilings $\T$ and $\T'$.

Let us now see how to recover an initial tiling from any $(T',N')$-tiling, once a tiling-simulation $\psi$ is defined. Given a poly-tiling $\T' \in \psi(\T)$, one can restore unambiguously the initial (valid) tiling $\T$. Indeed, if $\T'(x,y)$ denotes the polyomino whose reference tile is located at $(sx, sy)$, then $\T(x,y) = \varphi^{-1}(S)$, where $S = \varphi(\T(x,y))$ is the only image set of $\varphi$ containing $\T'(x,y)$. This tiling $\T$ must be valid (all glues matching) in order for $\psi$ to be a simulation, which is enforced by the last constraint of Definition~\ref{def:tilingsimu} (since in this case, $\chi_{\varphi}(\T') = \T'$).

Moreover, for any other tiling $\T' \in \Tils_{T',N'}$ such that there is no tiling $\T$ with $\T' \in \psi(\T)$, the function $\chi_{\varphi}$ allows to recover from this situation: as in the previous case, by construction of $\varphi$, there exists a unique $(T,N)$-tiling $\T$ such that $\chi_{\varphi}(\T') \in \psi(\T)$. This tiling $\T$ can be recovered as explained in the previous paragraph. From the last constraint in Definition~\ref{def:tilingsimu}, we also deduce that $\T$ is valid.

\smallskip

The simulation of a zipper by a ribbon is defined similarly. Let $\varphi : T \to \P(\Ribs_{T'})$ associate a tile $t$ with a set of unique finite $T'$-ribbons, , such that for all $t \in T$, for all $(P,r) \in \varphi(t)$, $(0,0) \in \dom(r)$, and such that if $t \ne t'$, for all $(P,r) \in \varphi(t)$, $(P',r') \in \varphi(t')$, $u \in \dom(r)$, and $u' \in \dom(r')$, then $r(u) \ne r'(u')$. Also let $\chi_{\varphi} : \Ribs_{T'} \to \Ribs_{T'}$, which associates any $T'$-ribbon with a poly-ribbon as previously, by removing incomplete and misaligned polyominoes.

\begin{definition}
  A \emph{zipper-simulation} of the tile system $T$ by the tile system $T'$ is a mapping $\psi : \Zips_{T,N} \to \P(\chi_{\varphi}(\Ribs_{T'}))$, such that $\psi$ associates any $(T,N)$-zipper $(P,r)$ where $r : \range(P) \to T$ with a non-empty subset of poly-ribbons (hence ribbons) from $\set{(P,r') \st r' : \range(P) \to \Ribs_{T'}, r'(x,y) \in \varphi(r(x,y))}$; with the additional constraint that for all $T'$-ribbons $(P',r')$, there exists some unique valid $(T,N)$-zipper $(P,r)$ such that $\chi_{\varphi}((P',r')) \in \psi((P,r))$.
\end{definition}

%The reference is not required in this latter definition, because zippers and ribbons being connected, there is no risk of misalignment.
%The recovery procedure $\psi^{-1}$ then only consists in step \ref{it:simu_inv_1} above.
The recovery procedure proceeds as previously, by turning a $T'$-ribbon into a poly-ribbon with $\chi_{\varphi}$, and recovering the zipper tiles with $\varphi^{-1}$.

\smallskip

Provided a tiling-simulation $\psi$ exists, we say that $\T'$ \emph{simulates} $\T$ if $\chi_{\varphi}(\T') \in \psi(\T)$. Similarly, for a zipper-simulation $\psi$, the ribbon $(P',r')$ \emph{simulates} the zipper $(P,r)$ if $\chi_{\varphi}((P',r')) \in \psi((P,r))$.
Finally, when no ambiguity is possible, we simply use the term \emph{simulation} to refer to the simulation of tilings or zippers.

%%% naive

\subsection{Total Tilings Constructions}\label{sec:naive}

In~\cite{AKKR02}, the authors use a construction which can be adapted to the simulation of total tilings. Their idea is to replace each arbitrary-neighborhood tiles by von Neumann-neighborhood tiles, these new tiles (called \emph{supertiles} to avoid misunderstandings) being groupings of several of the initial tiles, such that supertiles group all the tiles which belong to the neighborhood of the initial tile. This leads to the following statement.

\begin{proposition}\label{prop:total}
  Let $T$ be a tile system in arbitrary neighborhood $N$, with set of glues $X$. There exist a simulation $\psi$ and a tile system $T'$ in neighborhood $\NVN$, with set of glues $X'$, such that any total $(T,N)$-tiling can be simulated by a total $(T',\NVN)$-tiling.
\end{proposition}

\begin{proof}
  Let $\T$ be a $(T,N)$-tiling. We replace tiles from $T$ by new tiles from $T' = T^{\card{N}+1}$ called \emph{supertiles}, each of them encoding all the tiles from $N$ and the tile itself. The 4 von Neumann glues encode the neighbors located at $N \cap \set{(i,j) \st i \geq 0}$ (East), $N \cap \set{(i,j) \st i \leq 0}$ (West), $N \cap \set{(i,j) \st j \geq 0}$ (North), and $N \cap \set{(i,j) \st j \leq 0}$ (South). Therefore, two supertiles stick if the matching glues encode the same tiles, allowing to transmit information to distant neighbors. Figure~\ref{fig:supertiles} illustrates this construction in the case of Moore neighborhood $N_M$.
  In this figure, we use the following notation: given a position $(x, y) \in \Z^2$, let tile $t = \T(x,y) \in T$. The symbol $N(t)$ [resp., $S(t)$, $E(t)$, $W(t)$] represents the tile $\T(x,y+1)$ [resp., $\T(x,y-1)$, $\T(x+1,y)$, $\T(x-1,y)$]; and denote $XY = X \circ Y$ the composition of any of the functions $X,Y \in \set{N,S,E,W}$.
  \begin{figure}[!ht]
  \begin{center}
    \includegraphics{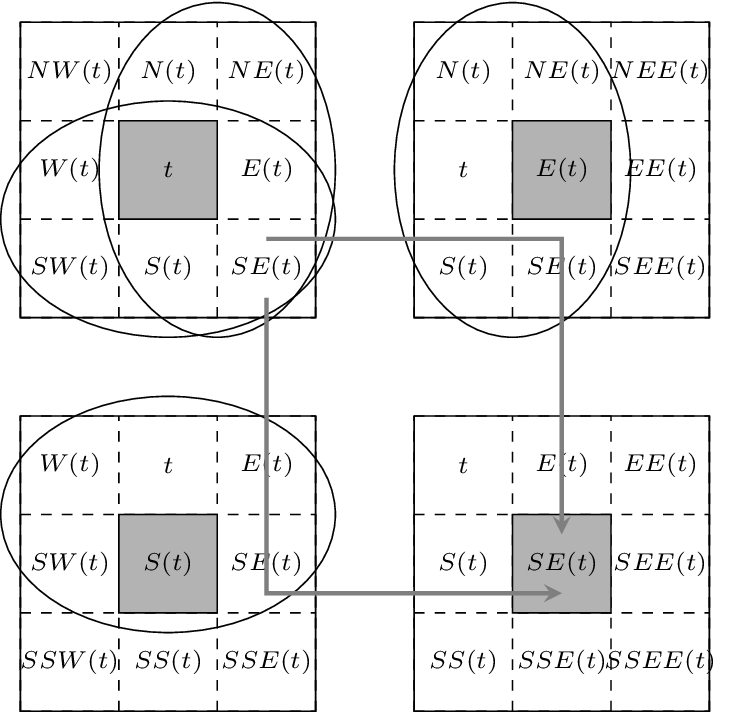}
    \caption{Supertiles assembly for simulating neighborhood $N_M$. The glues consist of the circled elements, they are used to transmit the information diagonally, as indicated by the arrows.}
    \label{fig:supertiles}
  \end{center}
  \end{figure}

In this case, the polyominoes simulating the original tiles are made of only one supertile. Besides, for each tile $t \in T$, the set $\varphi(t)$ contains all the one-tile tilings consisting of a supertile centered on $t$, with all the admissible neighbors of $t$. The function $\chi_{\varphi}$ is the identity, since all polyominoes are made of one supertile, they are necessarily complete and aligned. Any total tiling $\T \in \Tils_{T,N}$ is then simulated by an appropriate poly-tiling $\psi(\T)$, seen as a total $(T',\NVN)$-tiling $\T' \in \Tils_{T',\NVN}$.

The converse is trivial, any total $(T',\NVN)$-tiling can be translated back into a unique valid total $(T,N)$-tiling by selecting only the central symbol.
\qed\end{proof}

Remark that this technique can not be applied directly to the simulation of total zippers. Indeed, one would need to be able to transfer an arbitrary amount of information when simulating zippers such as the one represented in Fig.~\ref{fig:cexad}, since this amount would depend on the length of the zipper before it goes back.
\begin{figure}[!ht]
  \begin{center}
  \includegraphics{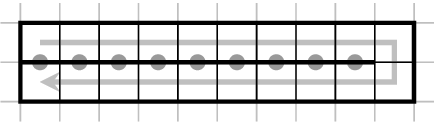}
  \caption{Von Neumann-neighborhood zipper, going backwards. The longer it is, the more information (represented by the circles) would need to be encoded in the supertiles.}
  \label{fig:cexad}
  \end{center}
\end{figure}

Also note that this simulation does not work for partial tilings, since in the absence of intermediate supertiles the information can not be transmitted. For example, in Fig.~\ref{fig:supertiles}, if the top-right and bottom-left supertiles are omitted, it is not possible to ensure that the bottom-right tile is centered on $SE(t)$.

This issue can be partially solved~\cite{W09}, by introducing a new symbol ``no tile'' in the supertiles, when a tile is absent. In this case, there will be supertiles at each position of the plane (possibly consisting only of symbols ``no tile''). However, this would cause every partial tiling to be simulated by a total tiling, which is not suitable in most applications.

%%%%%%%%%%%%%%%%%%%%%%%%%%%% CONSTRUCTIONS %%%%%%%%%%%%%%%%%%%%%%%%%%%%

\section{Simulating Arbitrary-Neighborhood Zippers by Ribbons}\label{sec:ziprib}

Here, we prove that there exists a simulation such that any zipper, using a tile system in arbitrary neighborhood, can be simulated by a ribbon using an appropriate tile system in von Neumann neighborhood. The objective of our constructions will be to define the function $\varphi$, in such a way that the polyominoes it produces are unique and do not overlap if and only if the initial tiles of the zipper stick.

The final construction being quite complex, we introduce the technical difficulties progressively. First, we recall known results which present the basic principles of the simulations, and solve the problem of crossings. Then, we deal with linear neighborhoods (all neighbors are on the same line), and finally we prove the general result in Corollary~\ref{cor:main}.

%%% preliminary (VN, Moore)

\subsection{Preliminary Results}\label{sec:VNM}

First we recall basic results of simulations of zippers.  In~\cite{AKKRS09}, the authors prove a fundamental result on the simulation of zippers by ribbons. They describe a method used to simulate bi-infinite zippers using ``directed'' tiles in von Neumann neighborhood by ribbons. The result can be extended to arbitrarily long zippers and standard tiles, as recalled here.

\begin{theorem}[\cite{AKKRS09}]\label{th:VN}
  Let $T$ be a tile system in neighborhood $\NVN$, with set of glues $X$. There exist a simulation $\psi$ and a tile system $T_{\mu}$ in neighborhood $\NVN$, with set of glues $X_{\mu}$, such that any $(T,\NVN)$-zipper can be simulated by a $T_{\mu}$-ribbon.
\end{theorem}

\begin{proof}
  The key of the proof is the construction of the simulation $\varphi$ which associates each tile with polyominoes. The basic idea behind these polyominoes is to replace every tile from $T$ by a unique shape (Figs.~\ref{fig:simuVN}\subref{fig:simuVN1} and~\ref{fig:simuVN}\subref{fig:simuVN2}). Glues are replaced by \emph{bumps} (the tile is raised) and \emph{dents} (the tile is dug), a different glue leading to a different bump or dent. A way to uniquely code the glues is to change the vertical or horizontal position of the bump or of the dent, depending on the glue. Then, like in a jigsaw puzzle, two shapes stick if their adjacent bump and dent fit into each other. Therefore, a ``ribbon'' of these shapes would simulate a $(T,\NVN)$-zipper, since the bumps and dents imply that the sides unconstrained by the ribbon have to match\ignore{ (for example in Fig.~\ref{fig:ribzip}\subref{fig:ribbon}, the bump $n$ would have to fit the dent $f$, enforcing $n=f$)}.
  
  The second step of the construction of $\varphi$ is the definition of a new tile system $T_{\mu}$, in von Neumann neighborhood and with a new set of glues, which will be used to build the polyominoes simulating the initial tiles from $T$. This leads to a path $P$, starting from the middle of the side corresponding to the input direction, and leaving at the middle of the side given by the output direction (Fig.~\ref{fig:simuVN}\subref{fig:simuVN3}). This path draws the contour of the shape, including bumps and dents, while the central part of the path is used to reconnect the input and output sides. Its position is determined by the point at the Southwest corner for example, which we choose as the reference tile located at $(0,0)$. The path is ``filled'' by new tiles from the set $T_{\mu}$ (we call these \emph{microtiles} to avoid misunderstandings) using a mapping $r : \range(P) \to T_{\mu}$. The finite $T_{\mu}$-tiled path $(P,r)$ is the polyomino associated with the initial tile. Remark that for one initial tile, there can be 12 different polyominoes, corresponding to the same shape but with 4 possible input and 3 possible output directions.
  
\begin{figure}[!ht]
  \subfloat[Tile of the original zipper with its four glues.]{\label{fig:simuVN1}
    \includegraphics{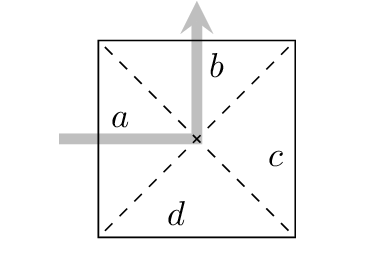}
}
  \hfill
  \subfloat[Coding of the glues into bumps and dents.]{\label{fig:simuVN2}
    \includegraphics{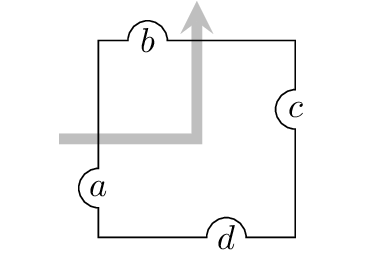}
}
  \hfill
  \subfloat[Polyomino used to simulate the original tile.]{\label{fig:simuVN3}
    \includegraphics{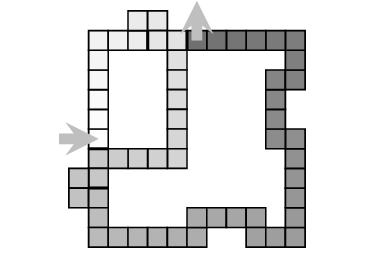}
}
  \caption{Simulation of a von Neumann-neighborhood zipper using a ribbon of microtiles. The three steps of the simulation of a single tile are represented, in each one the gray arrow indicates the direction of the underlying path of the zipper.}
  \label{fig:simuVN}
\end{figure}

  The polyomino should be a ribbon, so we should give some details about its construction. The first microtile (called the input microtile) has its input side colored with glue $(g,d)$, where $g$ is the input glue of the initial tile and $d$ the direction of the path (West-to-East, etc.); similarly the output side of the output microtile encodes the glue of the output side of the initial tile and the direction. For all other microtiles, the input matches the output of the previous tile, so that our polyomino $(P,r)$ is a $T_{\mu}$-ribbon. The input and output glues are unique among all the polyominoes, so this polyomino is the only possible ribbon using the tile system $T_{\mu}$, once the input microtile is given. Besides, to ensure that no interference occurs, the two sides which are not colored yet have a glue that matches nothing (for example glue $null_1$ on the West or North sides, $null_2$ on the East or South sides).

  For a tile $t\in T$, the set $\varphi(t)$ contains the 12 polyominoes described above. Then, obviously, if a $T$-tiled path is a $(T,\NVN)$-zipper, one can find a ``unique'' poly-ribbon consisting of the catenation of polyominoes constructed by $\varphi$, which is a $T_{\mu}$-ribbon. It is not really unique, since at the extremities of a finite zipper, 3 different polyominoes corresponding to the 3 possible input (or output) directions are admissible.
  
  Conversely, because of the careful design of the glues from $X_{\mu}$, any $T_{\mu}$-ribbon can be seen as a poly-ribbon, since the glues forming the polyominoes appear only once and guarantee that only polyominoes can be formed. A $T_{\mu}$-ribbon may have up to two incomplete polyominoes at the extremities, but if we cut them off using $\chi_{\varphi}$, then it can be represented as a unique poly-ribbon. Then, if a poly-ribbon using tile system $T_{\mu}$ exists, Definition~\ref{def:macrorib} implies that
\begin{itemize}
  \item all polyominoes do not overlap (condition~\ref{item:macrorib2}), hence that the glues they simulate match on the four sides;
  \item the polyominoes stick on their input/output tiles (condition~\ref{item:macrorib3}), hence that the polyominoes follow a path.
\end{itemize}
Using this path and these glues, one can uniquely restore the initial $(T,\NVN)$-zipper.
\qed\end{proof}

The following remark states an important property of our construction. In fact, the simulation $\psi$ we just constructed can be considered as bijective.

\begin{remark}\label{rem:bij}
  Let $\R$ be the equivalence relation which states that two $T_{\mu}$-ribbons are equivalent if they represent the same $(T,\NVN)$-zipper. Then, the simulation $\psi$ is a bijection between the set of $(T,\NVN)$-zippers $\Zips_{T,\NVN}$ and the set of $T_{\mu}$-ribbons quotiented by $\R$, $\Ribs_{T_{\mu}/\R}$.
\end{remark}

  In addition, the polyominoes used in the above simulation are rectilinear polyominoes, \ie a simple sequence of tiles outlining a shape. These are a particular case of general polyominoes, making the simulation more powerful.

\begin{remark}\label{rem:VNtotal}
  Theorem~\ref{th:VN} allows to transfer information on a ribbon in all 4 directions. This result, in conjunction with Proposition~\ref{prop:total}, allows a simple construction in the case of arbitrary-neighborhood total zippers, by first reducing to a von Neumann-neighborhood total zipper with the supertiles technique, and then simulating it by a two-tile-neighborhood ribbon.
\end{remark}

The simulation of a zipper in Moore neighborhood $N_M$ (\ie adding diagonal communications) is more complex, because diagonal glues cross, as illustrated in Fig.~\ref{fig:Mcross}.

\begin{figure}[!ht]
  \begin{center}
  \includegraphics{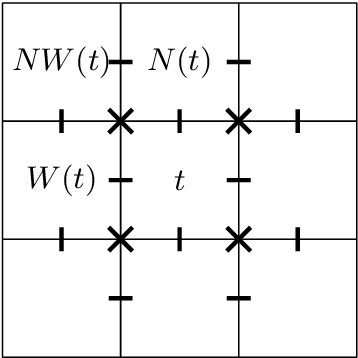}
  \caption{Communication between tiles in Moore neighborhood.} 
  \label{fig:Mcross}
  \end{center}
\end{figure}

The consequence is that diagonal bumps would also have to cross each other. This issue is solved in~\cite{CK10} using a method that we summarize here, because it will turn out to be useful in our constructions.

\begin{theorem}[\cite{CK10}]\label{th:moore}
  Let $T$ be a tile system in neighborhood $N_{M}$, with set of glues $X$. There exist a simulation $\psi$ and a tile system $T_{\mu}$ in neighborhood $\NVN$, with set of glues $X_{\mu}$, such that any $(T,N_M)$-zipper can be simulated by a $T_{\mu}$-ribbon.
\end{theorem}

\begin{proof}
  The global idea of the simulation is the same as for the proof of Theorem~\ref{th:VN}, we need to define a function $\varphi$ for all tiles from $T$. In Figs.~\ref{fig:moore}\subref{fig:moore1} and~\ref{fig:moore}\subref{fig:moore2} we present a picture of the shape that can be used to simulate an initial tile of $T$. It differs slightly from the shape described in~\cite{CK10}, but the idea is the same and this new shape is an introduction to our results.
  
  \begin{figure}[!ht]
    \begin{center}
      \subfloat[A single shape. Next to the communication places is written the neighbor it communicates with.]{\label{fig:moore1}
        \includegraphics{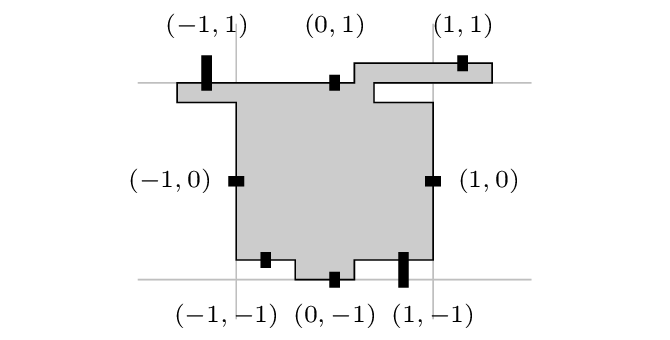}
      }
      \hfill
      \subfloat[The shape and its 8 neighbors.]{\label{fig:moore2}
        \includegraphics{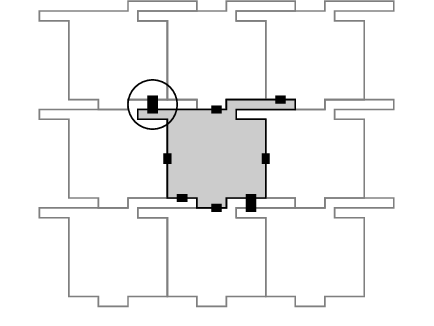}
      }
    \end{center}
    \caption{Shape used to simulate a tile in Moore neighborhood, filled with gray for clarity sake. The communication bumps and dents are represented by short thick lines.}
    \label{fig:moore}
  \end{figure}
  
  This shape is turned into a polyomino using new microtiles from $T_{\mu}$, taking into account the input and output directions. Then, all but one of the communications can be done as previously, by modifying the position of bumps and dents to simulate different glues. The only issue is the communications between the Northwest and Southeast neighbors (circled on Fig.~\ref{fig:moore}\subref{fig:moore2}). We have to be able to relay information about these diagonal glues without the physical touch between edges of tiles. According to the notation from Fig.~\ref{fig:Mcross}, we need to check the glue compatibility between the central tile $t$ and its Northwest neighbor $NW(t)$, as well as between the tiles $N(t)$ and $W(t)$.  This can be accomplished by a geometrical construction such as the one in Figs.~\ref{fig:crossing2}\subref{fig:cross1a} and~\ref{fig:crossing2}\subref{fig:cross1b}.

\begin{figure}[!ht]
  \begin{center}
%  \hfill
  \subfloat[Glues matching between $t$ and $NW(t)$.]{\label{fig:cross1a}
    \includegraphics{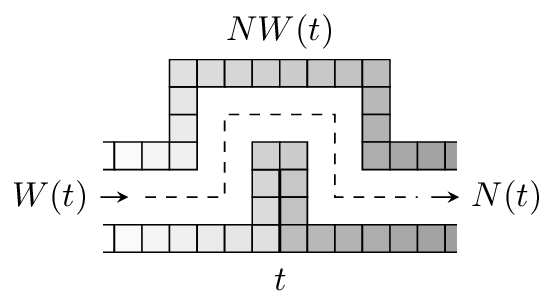}
  }
  \hfill
  \subfloat[Glues not matching: the ribbons overlap.]{\label{fig:cross1b}
    \includegraphics{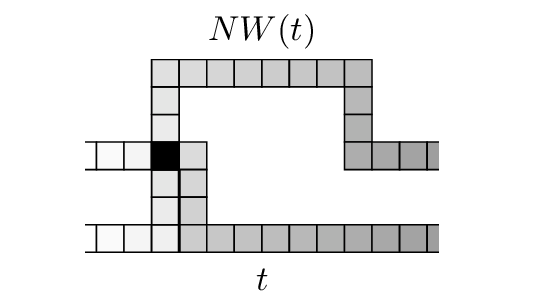}
  }
%  \hfill\null
  \caption{``Crossing'' of information. The top ribbon segment is part of $NW(t)$, the bottom ribbon segment is part of $t$, and the space in between can be filled by
2 layers of microtiles for $W(t)$ to actually reach $N(t)$.}
  \label{fig:crossing2}
  \end{center}  
\end{figure}

This construction checks the match between the $W(t)$ and $N(t)$ in the old-fashioned way, by physically matching the adjacent bump and dent of the corresponding polyominoes. The novelty of this construction is that the glue-match between $t$ and $NW(t)$ is accomplished without the respective polyominoes ever touching. Moreover, this construction works even in the case of partial tilings where the tile $W(t)$ might be missing. This is accomplished by carefully designing the shape and space between the bump and the dent so that, whether or not the tile between these polyominoes is present, their shapes will be compatible and not overlap if and only if the glues of the corresponding tiles were compatible. In order to achieve this and cross the 2 layers of microtiles forming the polyomino $W(t)$, the bump of $t$ should be 3 microtiles high, and the dent of $NW(t)$ should be 8 microtiles wide and 3 microtiles deep. In the sequel, the inner layers of $W(t)$ are called \emph{bridges}.
 
For any tile $t \in T$, if a polyomino in the set $\varphi(t)$ has a bridge, then $\varphi(t)$ should contain the same polyomino with the bridge at all possible locations, matching all possible glues. Indeed, the bridges do not participate in a communication, they are just a kind of ``relay'' which should be able to match any glue. Therefore, $\varphi(t)$ contains $\card{X}$ copies of each of the $12$ polyominoes given by the above construction, for the $\card{X}$ possible locations of the bridges (one per glue).

The end of the proof, which consists in restoring a $(T,N_M)$-zipper from any $T_{\mu}$-ribbon, is similar to the proof of Theorem~\ref{th:VN} and is left to the reader.
\qed\end{proof}

\begin{remark}\label{rem:adleman}
A simpler construction~\cite{A01} using ``pitcher-tiles'' (see Fig.~\ref{fig:pitcher}) solves the problem of information crossing when simulating Moore-neighborhood zippers by ribbons, but only when the zipper is a \emph{total} tiling. In that case, one could use (with the notation from Fig.~\ref{fig:pitcher}) the spike of the $W(t)$ polyomino which conveys information to $N(t)$ as ``information carrier'' to transmit information from $t$ to $NW(t)$. This construction does not work as such in the case of zippers which are partial tilings, because the ``carrier'' $W(t)$ tile might be altogether absent.

\begin{figure}[!ht]
  \begin{center}
  \includegraphics{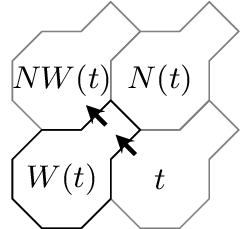}
  \caption{``Pitcher-tiles'' used to relay diagonal information, for total-tiling zippers only.}
  \label{fig:pitcher}
  \end{center}
\end{figure}  

This idea is similar to the one suggested in Remark~\ref{rem:VNtotal}, since it uses intermediate polyominoes to relay information. The main difference is that the polyomino-based simulation is more general and can be used for arbitrary neighborhoods, while these pitcher-tiles were designed specifically for the Moore neighborhood.
\end{remark}

%%% linear

\subsection{Simulating Arbitrary Linear Neighborhoods}\label{sec:linear}

We call \emph{linear} neighborhood a neighborhood $N$ such that if $(i,j) \in N$ then $j=0$. Although this is a sub-case of the general case studied in the next section, we will explain it in detail since it lays the base of the general study, and it is much simpler to describe and understand. Also note that this result was already announced in~\cite{CK09}, but the polyomino used there was not easily generalizable to arbitrary neighborhoods. First, we state an initial remark which allows us to consider only \emph{connected} linear neighborhoods $N_n = \set{(i,0) \in \Z^2 \st 0 < \abs{i} \leq n}$.

\begin{remark}\label{rem:1d}
  For a given set of glues $X$, any tile system $T$ defined in linear neighborhood $N$ can be replaced by an equivalent tile system $T'$ in an appropriate connected linear neighborhood $N_n$. Indeed, let $n$ be such that $N \subset N_n$, let $g \in X$ be an arbitrary ``dummy'' glue and $T' = \zeta_g(T)$ a tile system in neighborhood $N_n$, where $\zeta_g : T \to T'$ is defined for all $t \in T$ and $(i,j) \in N_n$ by
  \[
    \zeta_g(t)_{i,j} = \begin{cases}
      t_{i,j} & \text{if $(i,j) \in N$,}\\
      g & \text{otherwise.}
    \end{cases}
  \]
Then, clearly, $\T : D \to T$ is a $(T,N)$-tiling if and only if $\T' : D \to T'$ defined by $\T'(x,y) = \zeta_g(\T(x,y))$ is a $(T',N_n)$-tiling.
\end{remark}

First, we state a lemma which generalizes to an arbitrary number of inner layers the crossing operation detailed in the proof of Theorem~\ref{th:moore}. The gadget introduced in this lemma will be helpful in the next constructions.

\begin{lemma}\label{lem:crossing}
  In order to fit a bump and a dent spaced by $k$ layers, the bump must be $k+1$ microtiles high, the dent must be $2k+4$ microtiles wide and $k+1$ microtiles deep.
\end{lemma}

\begin{proof}
  The result is obtained by an immediate generalization of Fig.~\ref{fig:crossing2}\subref{fig:cross1a} to~$k$ white layers instead of~$2$.
\qed\end{proof}

We now prove the main result of this section for neighborhoods $N_n$.

\begin{theorem}\label{th:1D}
  Let $T$ be a tile system in connected linear neighborhood $N_n$, with set of glues $X$. There exist a simulation $\psi$ and a tile system $T_{\mu}$ in neighborhood $\NVN$, with set of glues $X_{\mu}$, such that any $(T,N_n)$-zipper can be simulated by a $T_{\mu}$-ribbon.
\end{theorem}

\begin{proof}
  As previously, we are going to replace a tile of the $(T,N_n)$-zipper by a shape using a function $\varphi$, leading to a set of polyominoes. The general shape of the polyominoes is illustrated in Fig.~\ref{fig:tuile1D}. It consists of a spike sent from the original square to the neighbor at distance $n$.
  
\begin{figure}[!ht]
  \begin{center}
  \includegraphics{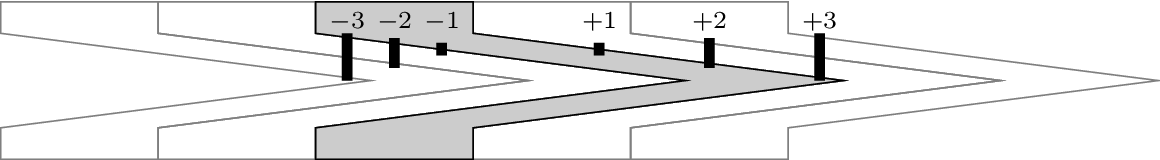}
  \caption{General shape simulating a tile in arbitrary linear neighborhood. In this example, the neighborhood is $N_3 = \{(-3,0)$, $(-2,0)$, $(-1,0)$, $(1,0)$, $(2,0)$, $(3,0)\}$. The shape is drawn in thin black, and the vertical thick lines are the communication places with the neighbor whose abscissa is the number indicated above it. These neighbors are drawn in gray.}
  \label{fig:tuile1D}
  \end{center}
\end{figure}

As shown on the picture, communication bumps can be put on the spike, while dents are located inside the initial square. For matching the glues between one tile and its neighbor $(i,0)$, the bump will cross $i-1$ other spikes, we will see later how many layers of microtiles it represents. Indeed, the essential part of the simulation is the discretization of this shape into a polyomino of microtiles. A final polyomino is represented in Fig.~\ref{fig:1Ddetail}.

\begin{figure}[!ht]
  \begin{center}
    \includegraphics{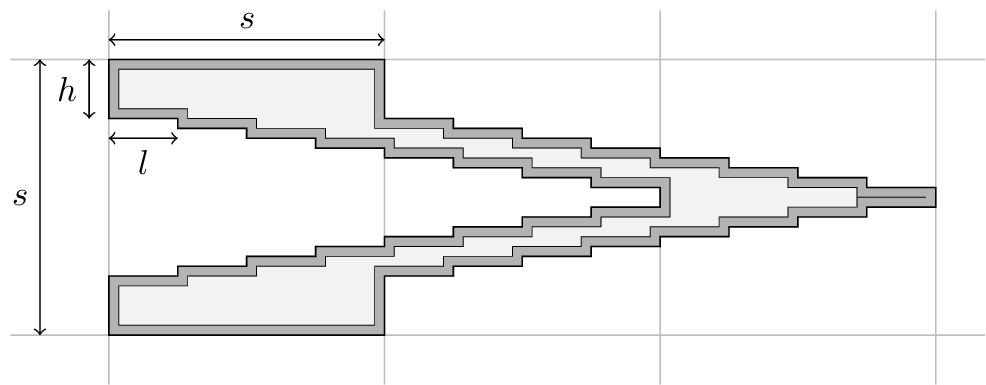}
    \caption{Detail of the polyomino used for arbitrary linear-neighborhood zippers. Here, $n=2$, $l = 7$, $s = 28$ and $h = 6$. The dark gray contour is filled by microtiles, the light gray area is filled for clarity and does not necessarily contain microtiles.}
    \label{fig:1Ddetail}
  \end{center}
\end{figure}

We define the following variables. The height of the polyomino is $s$ microtiles, it is also the width of the initial tile and therefore the scale of the final poly-ribbon. The vertical space between the top of the polyomino and the beginning of the spike is denoted $h$, and since the spike is centered, the space between the end of the spike and the bottom is also $h$. Finally, the spike is a succession of horizontal segments of $l$ microtiles, thus the slope of the spike is $\pm1/l$. Then, the following constraints have to be respected when constructing the polyomino.
\begin{itemize}
  \item The integer $l$ has to be as large as necessary, leaving enough horizontal space for crossings, as suggested by Lemma~\ref{lem:crossing}. Similarly, $h$ needs to provide enough vertical space for the dents. Formally, $l \geq \alpha$ and $h \geq \beta$, where $\alpha$ and $\beta$ will be given later on.
  \item The polyomino must be at least 3 layers wide everywhere, two for the inner and outer layers of the spike, and one for a potential junction of the path from the input microtile to the output microtile (as in Fig.~\ref{fig:simuVN}\subref{fig:simuVN3}). As a consequence, $h \geq 3$ and $s \geq 3l + 3$, since the inner layer of the spike must go down by $3$ before reaching $s-3$. The first constraint can be removed (assuming $\beta \geq 3$), while the second one can be replaced by $s \geq 4l$ (provided $l \geq 3$, which is the case if $\alpha \geq 3$), so that $s$ can be exactly divided in 4 horizontal segments everywhere on the spike.
  \item The initial tile being a square, the height is $s$ and can also be written $h + 2\lfloor ns/l\rfloor + h$ (the spike goes $ns$ microtiles to the right at slope $1/l$). Therefore, after replacing both $s$ by $4l$, we have $h = 2l - 4n$.
\end{itemize}
Since $h \geq \beta$, because of the last equation $l$ has to be greater than $2n+\beta/2$. For given $n, \alpha, \beta \in \N$, a solution of the system is then
\[
\left\{\begin{array}{l}
  l = \lceil \max(\alpha, 2n+\beta/2) \rceil \\
  s = 4l \\
  h = 2l - 4n\enspace.
\end{array}\right.
\]
It is quite obvious that translated copies of this polyomino tile the plane with no overlaps, allowing to replace a grid of tiles by a grid of polyominoes\footnote{A formal proof of this fact could be made using the characterization of polyominoes tiling the plane given in~\cite{BN91}. Indeed, the contour word of our polyominoes (without bumps and dents, and filled to match the definition in~\cite{BN91}) would prove them to be pseudo-squares.}.

\smallskip

We now give some details on how the communications take place. For a more convenient description, we split the polyomino into $n+1$ horizontal parts of width $s = 4l$, we denote them from left to right by part~$0$ (which corresponds to the initial square tile) to part~$n$ (end of the spike). Each of these parts is divided into~$4$ horizontal segments of length~$l$, denoted segment~$1$ to segment~$4$. As suggested in Fig.~\ref{fig:tuile1D}, the bumps and bridges are located on the horizontal segments on the top of the spike in parts~$1$ to~$n$, while the corresponding dents are on the horizontal steps on the top of the hole in part~$0$. For every $0 < i \leq n$, the glue $t_{i,0}$ of the initial tile is allocated some space in every part, on one of the four segments. A simple way to do this is to allocate glue $t_{i,0}$ to segment $(i-1 \mod 4)$, hence each segment is used for $\lfloor n/4\rfloor$ or $\lceil n/4\rceil$ glues. This space is then used in part~$0$ for a dent, in parts~$1$ to~$i-1$ for bridges, and in part~$i$ for the bump.

Moreover, there are $i-1$ spikes to cross by the bump encoding $t_{i,0}$. Each spike is~$4$ layers thick, hence there are at most $4(n-1)$ layers to cross. This number is fixed, so we can apply Lemma~\ref{lem:crossing} with $k = 4(n-1)$: each glue needs a horizontal space of $(8(n-1)+4) \text{ (width of a dent)} + (\card{X}-1)(4(n-1)+1)\text{ (spacing between different possible glues)}$ microtiles. This gives a lower bound to $l$, which has to be greater than $\lceil n/4\rceil \times (\card{X}(4n-3) + 4n - 1)$. After adding $6$ microtiles to separate the dents from the sides of the polyomino, we define
\[
  \alpha = \lceil n/4\rceil \times (\card{X}(4n-3) + 4n - 1) + 6\enspace,
\]
as the lower bound for $l$ used previously. On the other hand, the depth of the dents is $k+1= 4n-3$. Consequently, we have another constraint on $h$ which must be greater than $\beta = 4n$ (space for the biggest dent plus the three original layers). This means $l$ greater than $4n$, which is already the case because $l \geq \alpha \geq 4n$. Then,
\[
\left\{\begin{array}{l}
  l = \lceil n/4\rceil \times (\card{X}(4n-3) + 4n - 1) + 6 \\
  s = 4l \\
  h = 2l - 4n\enspace.
\end{array}\right.
\]

This ensures that our polyomino can be constructed, using an appropriate set of microtiles which will generate the $T_{\mu}$-ribbon we described. The last step of the construction is the positioning of the input and output microtiles. We can choose to place them at top-left position (path coming from or going to the West), top-right (path from or to the East), middle of the top side (path from or to the North), middle of the bottom side (path from or to the South). Since $s$ is even, the middle of a side is chosen after $\lfloor s/2\rfloor$ microtiles. Then the inner layer we preserved can be used for joining the input and the output microtiles easily; for example, starting from the input microtile, one can draw the contour of the path from the left, just before reaching the output microtile, the path goes one layer inside and goes back to the input microtile where it draws the contour from the right (see Fig.~\ref{fig:jonction} for examples).

\begin{figure}[!ht]
\begin{center}
  \includegraphics{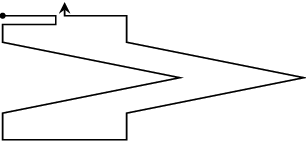}
  \hfill
  \includegraphics{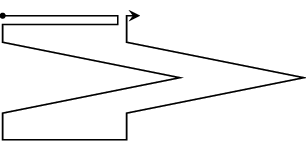}
  \hfill
  \includegraphics{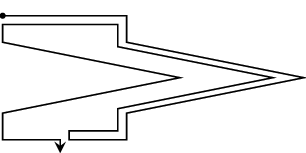}
\end{center}
  \vspace{-2ex}
  \caption{The three different paths when the input direction is West.}
  \label{fig:jonction}
\end{figure}

Finally, putting the appropriate polyominoes one after the other generates a poly-ribbon which is in fact a $T_{\mu}$-ribbon. As it was shown before, any $T_{\mu}$-ribbon can be transformed into a poly-ribbon with $\chi_{\varphi}$, from which one can restore the original $(T,N_n)$-zipper using $\varphi^{-1}$. The poly-ribbon does not overlap if and only if the glues from the $(T,N_n)$-zipper match everywhere, hence $\psi$ is a simulation.
\qed\end{proof}

\subsection{Complexity of the Linear-Neighborhood Construction}\label{sec:nbtuiles1d}

We now give results on the ``size'' of this simulation, which underline the fact that the generated polyominoes can be very complex. All the results refer to the construction described in the proof of Theorem~\ref{th:1D}.

\begin{lemma}\label{lem:bridges}
  A polyomino used to simulate a tile of a $(T,N_n)$-zipper contains $B = \mathcal{O}(n^2)$ bridges.
\end{lemma}

\begin{proof}
  When $n=2$, there is one bridge between the neighbors at $(-1,0)$ and at $(1,0)$; when $n=3$, there are in addition two bridges between neighbors $(-2,0)$ and $(1,0)$, and $(-1,0)$ and $(2,0)$. In general, there are $n-1$ bridges between $(i-n,0)$ and $(i,0)$ for $1 \leq i \leq n-1$, plus $n-2$ bridges between $(i-n-1,0)$ and $(i+n-1,0)$ for $1 \leq i \leq n-2$, and so on. Hence there are $B=\sum_{i=1}^{n-1} i=\frac{1}{2}n(n-1)$ bridges.
\qed\end{proof}

\begin{lemma}\label{lem:nbmicro1D}
  A polyomino used to simulate a tile of a $(T,N_n)$-zipper is constituted by $\mathcal{O}(\card{X}n^3)$ microtiles.
\end{lemma}

\begin{proof}
  Drawing the contour of a shape requires:
  \begin{itemize}
    \item $s$ microtiles for the~$2$ horizontal segments in part~$1$;
    \item $4h$ microtiles for all~$4$ vertical portions;
    \item $4(ns+4n)$ microtiles for the 2 spikes ($ns$ for the horizontal segments, $4n$ for the steps down and up);
    \item at most $s + 2h + 2(ns+4n) + s/2$ microtiles for joining input and output microtiles (worst case when joining left to bottom);
    \item $2n$ bumps and dents of height at most $\mathcal{O}(n)$ microtiles;
    \item $B=\mathcal{O}(n^2)$ (Lemma~\ref{lem:bridges}) bridges of height at most $\mathcal{O}(n)$ microtiles, each one at most~$3$ layers thick.
  \end{itemize}
  Since $s = 4l$ and we can choose $l = \lceil \alpha \rceil = \mathcal{O}(\card{X}n^2)$, after summing all of the above we obtain the result.
\qed\end{proof}

\begin{lemma}\label{lem:nbpolyominoes1D}
  Every tile of the initial $(T,N_n)$-zipper is simulated by $\mathcal{O}(\card{X}n^2)$ different polyominoes.
\end{lemma}

\begin{proof}
  For each tile $t \in T$, there are $4$ (number of input positions)${}\times 3$ (number of output positions)${}\times B\card{X}$ (number of different possible bridges) different paths, where $B$ is the number of bridges on the path. Since $B = \mathcal{O}(n^2)$ (Lemma~\ref{lem:bridges}), there are $\mathcal{O}(\card{X}n^2)$ different polyominoes for one tile when $n \geq 2$. When $n = 1$, there are no bridges and there are only 12 different polyominoes.
\qed\end{proof}

\begin{proposition}\label{prop:nbtuiles1d}
  In our construction, a $T_{\mu}$-ribbon simulating a $(T,N_n)$-zipper needs $\card{T_{\mu}} = \mathcal{O}(\card{T}\cdot\card{X}^2n^5)$ microtiles and $\card{X_{\mu}} = \mathcal{O}(\card{T}\cdot\card{X}^2n^5)$ glues.
\end{proposition}

\begin{proof}
  A $(T,N_n)$-zipper can use at most $\card{T}$ tiles, according to our construction each of them is simulated by $\mathcal{O}(\card{X}n^2)$ different polyominoes (Lemma~\ref{lem:nbpolyominoes1D}) constituted by $\mathcal{O}(\card{X}n^3)$ unique microtiles (Lemma~\ref{lem:nbmicro1D}). Hence the simulation needs $\card{T_{\mu}} = \mathcal{O}(\card{T}\cdot\card{X}^2n^5)$ different microtiles.
  
  Since a polyomino is a ribbon which has to be uniquely built, microtiles (except for the input and output ones) have 2 glues which can be found only on one other microtile. Therefore each of these microtiles introduce a new glue, and reuse another: there are at least $\card{X_{\mu}} = \mathcal{O}(\card{T_{\mu}}) = \mathcal{O}(\card{T}\cdot\card{X}^2n^5)$ glues. The other two glues are $null_1$ and $null_2$, and the input and output glues microtiles use glues from $X$, which does not change the order of magnitude of $\card{X_{\mu}}$.
\qed\end{proof}

%Remark that if $T$ contains all possible tiles, then $\card{T} = \card{X}^{\card{N_n}} = \card{X}^{2n}$. Although in general $T$ will contain a much smaller amount of tiles, we can not claim that the number of microtiles and glues used to simulate a $(T,N_n)$-zipper is polynomial in~$n$.

%%% arbitrary

\subsection{Simulating Arbitrary Neighborhoods}\label{sec:arb}

In this section, we prove the first of the two main results of this paper, namely that zippers in any neighborhood can be simulated by ribbons (Corollary~\ref{cor:main}). First, note that in a similar way to Remark~\ref{rem:1d}, any neighborhood $N$ can be replaced by an equivalent \emph{rectangular} neighborhood $N_{m,n} = \{(i,j) \in \Z^2 \st 0 \leq \abs{i} \leq m,\: 0 \leq \abs{j} \leq n \text{ and } (i,j) \ne (0,0)\}$ containing $N$. 

\begin{theorem}\label{th:rectangular}
  Let $T$ be a tile system in rectangular neighborhood $N_{m,n}$, with set of glues $X$. There exist a simulation $\psi$ and a tile system $T_{\mu}$ in neighborhood $\NVN$, with set of glues $X_{\mu}$, such that any $(T,N_{m,n})$-zipper can be simulated by a $T_{\mu}$-ribbon.
\end{theorem}

\begin{proof}
  The key of the proof is the generalization of the $\varphi$ simulation from the proof of Theorem~\ref{th:1D} to rectangular neighborhoods. The idea is to have a vertical succession of $n+1$ spikes of length $m$, each of them being a ``sheath'' for the next one (Fig.~\ref{fig:general}).
\begin{figure}[!ht]
  \begin{center}
    \includegraphics{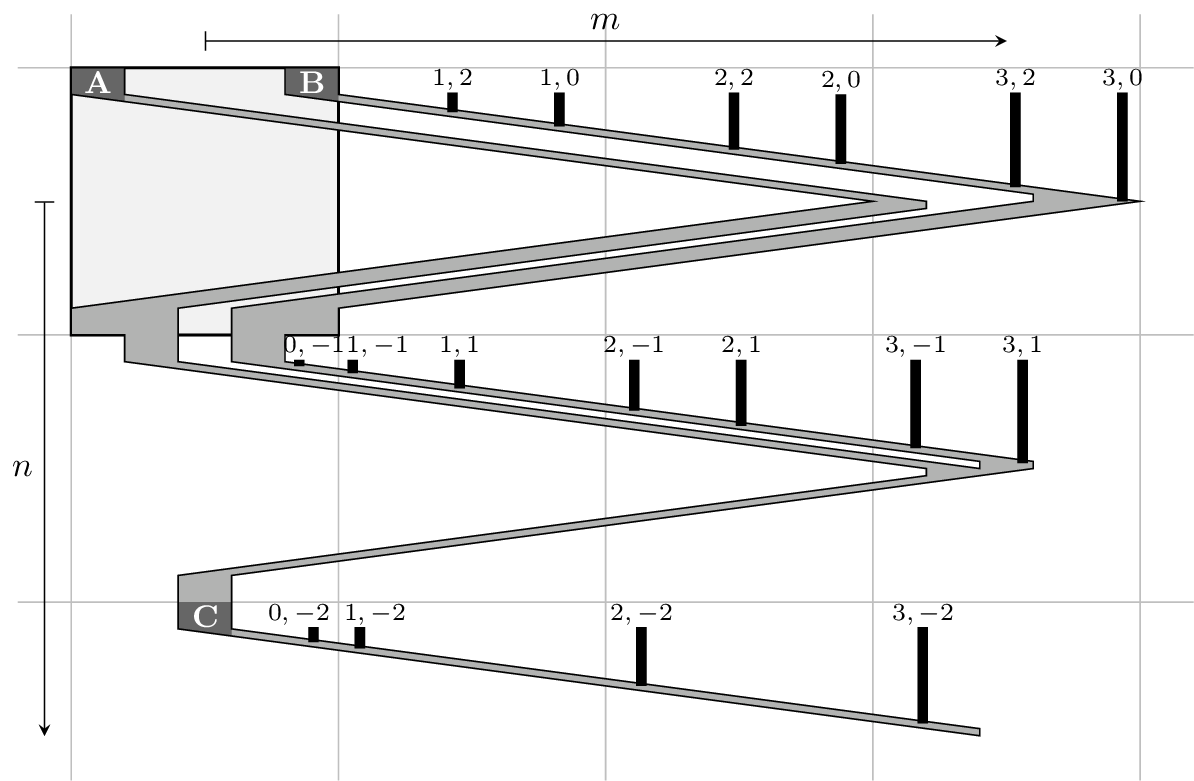}
    \caption{Shape simulating a tile defined in a rectangular neighborhood $N_{m,n}$ of size $(2m+1) \times (2n+1)$ (here $m=3$ and $n=2$). The initial tile is the light gray square, while the path is filled with darker gray.}
    \label{fig:general}
   \end{center}
\end{figure}

The bumps used for communications are put all along the spikes, at the places indicated on the picture. The difficulty is to find places for the dents, so that they can be contained in a place where enough space can be reserved. This is achieved as follows.
  \begin{itemize}
    \item North (from $(0,j)$, for $0 < j \leq n$) and West (from $(i,0)$, for $-m \leq i < 0$) communications are received in the area marked~$\mathbf{B}$ on the picture. This works very similarly to the linear-neighborhood case.
    \item Northwest communications (from $(i,j)$, for $-m \leq i < 0$ and $0 < j \leq n$) are received in~$\mathbf{A}$. Again, it is not difficult to see how the information crosses the layers.
    \item Southwest communications (from $(i,j)$, for $-m \leq i < 0$ and $-n \leq j < 0$) are received in~$\mathbf{C}$. This is slightly more complicated to understand, since these dents are not located inside the initial square. Putting the dent at the bottom of the polyomino allows to simulate the Southwest communications by Northwest communications, which is then easily achieved by positioning bumps on the spike.
  \end{itemize}
The key point is that the areas marked~$\mathbf{A}$,~$\mathbf{B}$, and~$\mathbf{C}$ (in dark gray on Fig.~\ref{fig:general}) are scalable, both vertically and horizontally, so they can be made as big as needed for the dents. Indeed, we can denote as previously by $s$ the width and height of the initial tile, hence the scale of the poly-ribbon. Let $x$ be the width of~$\mathbf{A}$,~$\mathbf{B}$,~$\mathbf{C}$ and of all the other horizontal subdivisions of $s$. Since there are two of these blocks for each of the first $n$ vertical spikes, plus one for the last spike, it holds that $s = (2n+1)x$. As in Theorem~\ref{th:1D}, we need the spike to be~$3$ layers wide, hence the slope $1/l$ of the spikes is defined by $l$ such that $x \geq 3l +3$. Again it is possible to decide that $x = 4l$, \ie a block is made of four descending horizontal segments. Finally, let $h$ be the left height of the $\mathbf{A}$-$\mathbf{B}$-$\mathbf{C}$ areas. The fact that the initial tile is a square is expressed by $s = 2h + 2\lfloor ms/l\rfloor$. We have the following equations:
\[
\left\{\begin{array}{l}
  x = 4l\\
  s = (2n+1)x\\
  s = 2h + 2\lfloor ms/l\rfloor\enspace.
\end{array}\right.
\]
Besides, to ensure enough space for the bumps and dents, we want to make sure that $x$ and $h$ are as big as necessary, \ie $x \geq \alpha$ and $h \geq \beta$ (the exact values of $\alpha$ and $\beta$ will be given later). Once solved, the system gives $h = (2n+1)(2l-4m)$. Since $h$ should be greater than $\beta$ and $x$ greater than $\alpha$, $l$ needs to be greater than $\max(\alpha/4, \beta/(4n+2) + 2m)$. Thus, a solution to the system which guarantees that the polyomino can be constructed is
\[
\left\{\begin{array}{l}
  l = \lceil \max(\alpha/4, \beta/(4n+2) + 2m) \rceil\\
  x = 4l\\
  h = (2n+1)(2l-4m)\\
  s = 4(2n+1)l\enspace.
\end{array}\right.
\]

\smallskip

A last technical difficulty is the description of how the spikes shrink to fit into the previous one. The general way to do this is illustrated on Fig.~\ref{fig:shrink}, which zooms on the rightmost part of the first spike of a polyomino.

\begin{figure}[!ht]
\begin{center}
  \includegraphics{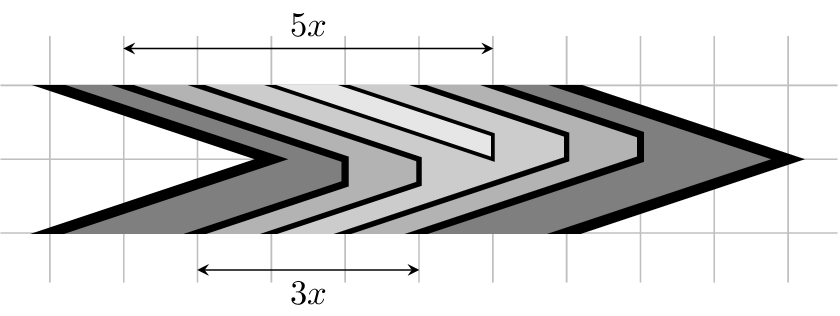}
\end{center}
\vspace{-2ex}
\caption{Illustration of how spikes shrink from width $ax$ to $(a-2)x$. The darkest gray corresponds to the spike of the initial tile, lighter grays represent the spike of farther North neighbors.}
\label{fig:shrink}
\end{figure}

The outer spike does not shrink because it does not need to, since the shrinking will happen at the bottom of the initial tile (see Fig.~\ref{fig:general}). This is not necessary, but it allows a tiling of the plane without any holes between polyominoes. The other spikes shrink by $x$ microtiles on the left and on the right by going down $4$ tiles (remember that the slope of the spikes is $1/l = 4/x$). Remark the slight asymmetry of the shrinking, which has to take place after reaching height $s/2$ on the left, and just before on the right. With all these conditions respected, it should be clear that vertically and horizontally translated copies of this polyomino tile the plane with no overlaps.

\smallskip

It remains to determine the bounds $\alpha$ and $\beta$. An application of Lemma~\ref{lem:crossing} to all the spikes and dents gives us the maximal size of bumps and dents. It is obtained for the diagonal communication between a tile and its neighbor located at $(m,-n)$, which crosses $n + (m-1)(2n+1)$ spikes of thickness $4$ layers, hence a total number of $k=4(2mn+m-n-1)$ layers. It follows that $\beta = k+3$ to ensure a free layer between the top of a dent and the upper side of the polyomino. The bound $\alpha$ is more complicated to express. As for Theorem~\ref{th:1D}, a communication bump-dent needs $2k+4$ horizontal space, plus an extra $(\card{X}-1)(k+1)$ microtiles for the different glues. The size $x$ has to allow $m+n$ dents in~$\mathbf{B}$, $mn$ dents in~$\mathbf{A}$ and~$\mathbf{C}$. Hence, after adding $6$ microtiles for preserving the borders of the block, $\alpha = \max(mn,m+n)\times (2k+4+(\card{X}-1)(k+1)) + 6$, with $k=4(2mn+m-n-1)$.

The rest of the proof (positioning and joining the input and output microtiles, bijection between a $(T,N_{m,n})$-zipper and a set of unique poly-ribbons) is unchanged from the proof of Theorem~\ref{th:1D}.
\qed\end{proof}

\begin{corollary}\label{cor:main}
  Let $T$ be a tile system in arbitrary neighborhood $N$, with set of glues $X$. There exist a simulation $\psi$ and a tile system $T_{\mu}$ in neighborhood $\NVN$, with set of glues $X_{\mu}$, such that any $(T,N)$-zipper can be simulated by a $T_{\mu}$-ribbon.
\end{corollary}

Remark~\ref{rem:bij} can be extended to the general case, hence any $T_{\mu}$-ribbon represents a unique $(T,N)$-zipper. Also note that a result similar to Proposition~\ref{prop:nbtuiles1d} could be stated, but it would be more complex and of little interest since the number of tiles would be a lot bigger.

To illustrate the achievability (and the complexity) of this construction, Fig.~\ref{fig:exemple1D} gives a complete example of a polyomino for the simulation of a linear-neighborhood $(T,N_2)$-zipper, with $\card{X} = 2$.

\begin{figure}[!p]
  \begin{center}
    \includegraphics{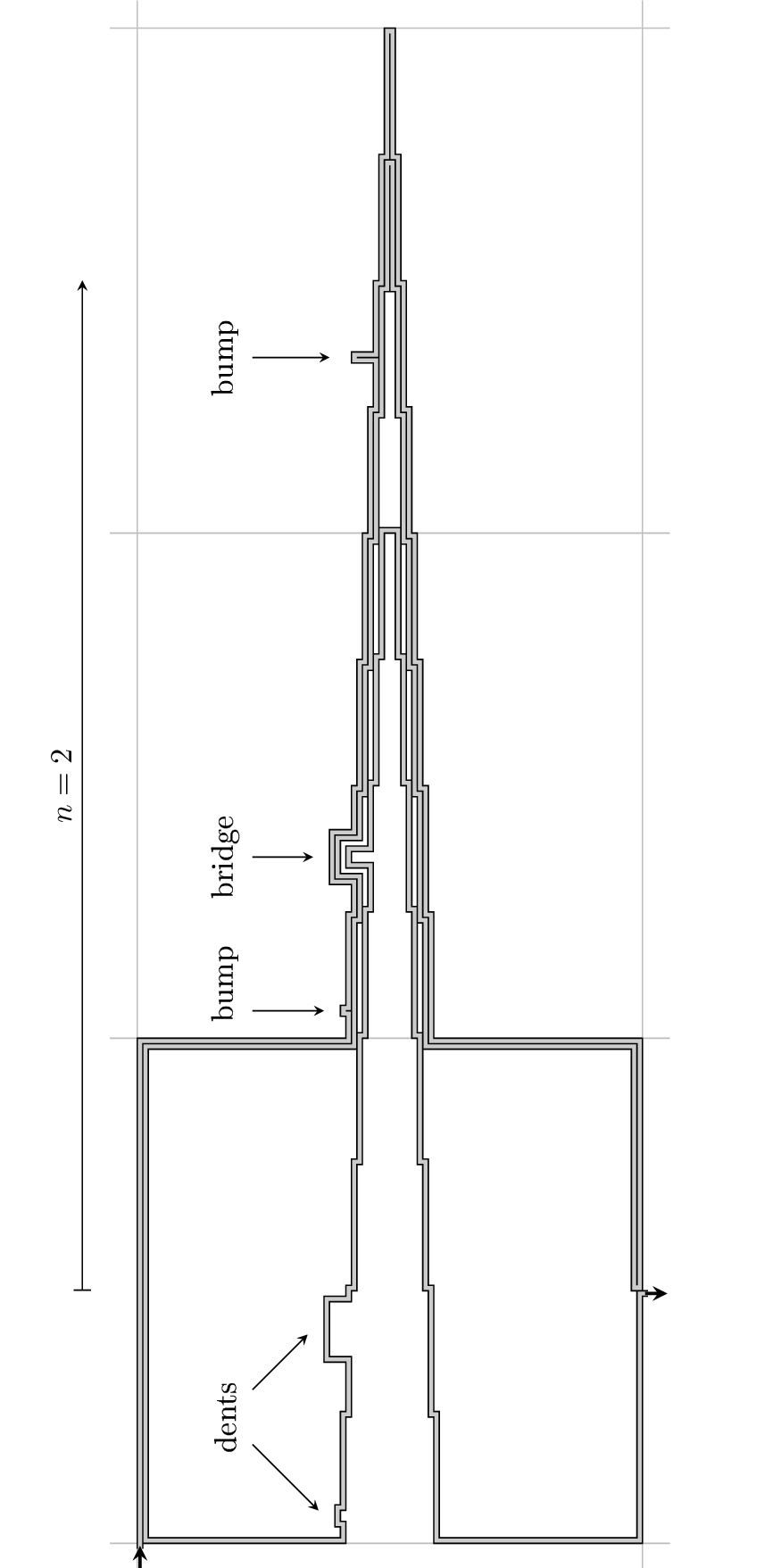}
  \end{center}
  \vspace{-2ex}
  \caption{Rotated polyomino simulating a tile of a $(T,N_2)$-zipper, with $\card{X}=2$, $l=23$, $L=92$, $h=38$, input direction is West and output direction is South. The microtiles are located in the gray layer, the input and output directions are indicated by the black arrows.}
  \label{fig:exemple1D}
\end{figure}

The general case is by far more complicated and a detailed picture would be difficult to understand. We put in Fig.~\ref{fig:exemple} the shape from Fig.~\ref{fig:general}, surrounded by 8 other shapes (colored alternatively in light and medium gray). The communication places between the central shape simulating the tile located at $(0,0)$ and the 8 other shapes are shown, one of them is detailed at the bottom of the picture.

\begin{figure}[!p]
  \begin{center}
    \includegraphics{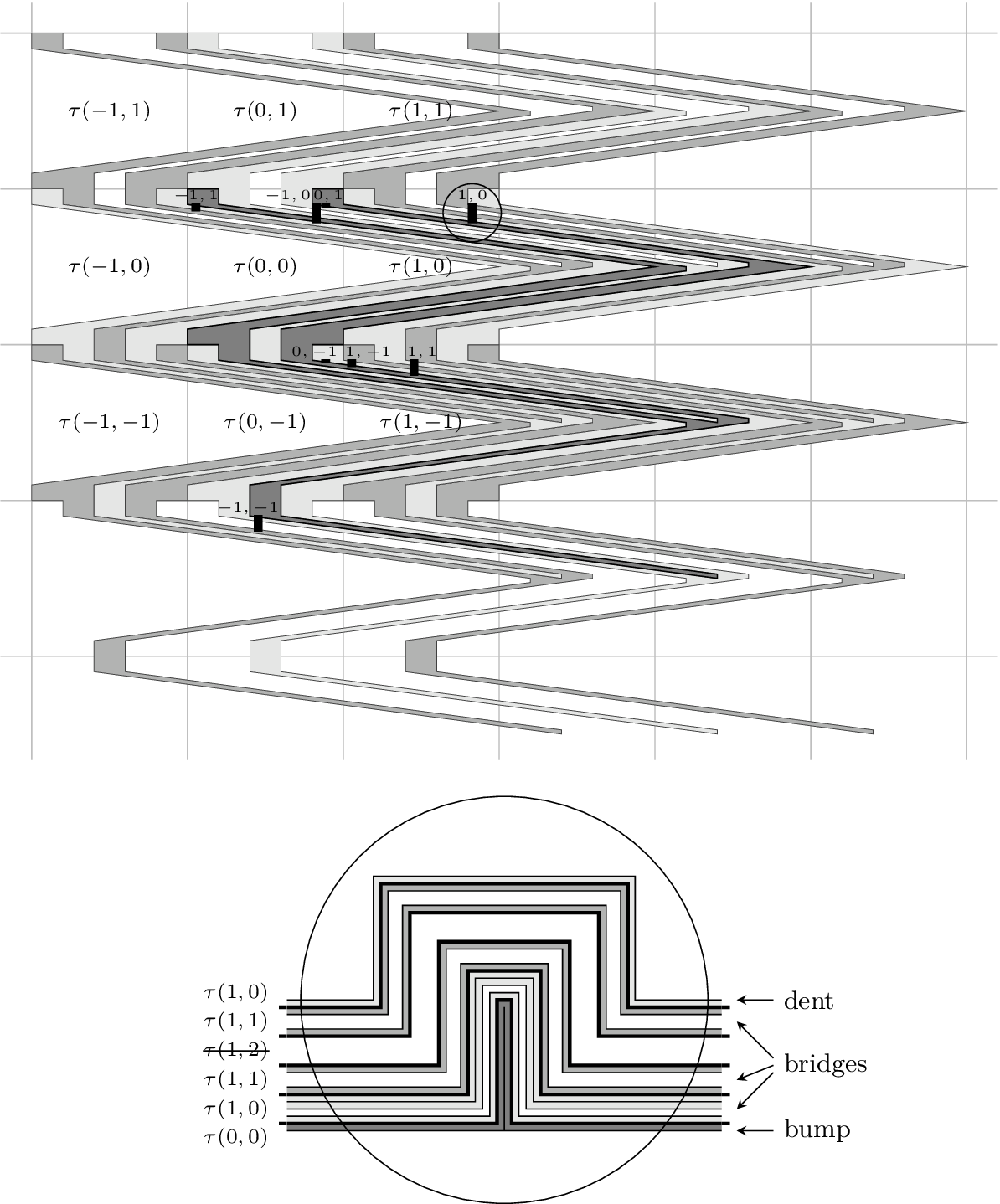}
  \end{center}
  \caption{In dark gray is represented the shape from Fig.~\ref{fig:general}, simulating tile $\T(0,0)$. It is surrounded by 8 other shapes corresponding to 8 neighboring tiles. Below is a detailed view of the circled communication place, with a bump crossing 4 bridges of width 4. Since the tile at position $(1,2)$ is absent, one bridge is missing. The other bridges contain either 2 or 3 layers of microtiles, depending on how the input and output microtiles are joined.}
  \label{fig:exemple}
\end{figure}

%%%%%%%%%%%%%%%%%%%%%%%%%%%% TILINGS %%%%%%%%%%%%%%%%%%%%%%%%%%%%

\section{Extension to Arbitrary-Neighborhood Tilings}\label{sec:tilings}

We now explain how the above construction can be modified, in order to simulate arbitrary-neighborhood tilings by von Neumann-neighborhood poly-tilings, to obtain our second main result (Corollary~\ref{cor:mainT}). Before that, we describe a simpler construction used in~\cite{AKKR02}, which can be adapted to do such a simulation, but only in the case of total tilings. We finally study some properties of the tilings which are preserved by our simulation technique.

\subsection{Polyomino Construction}

The construction described in Sect.~\ref{sec:ziprib} transforms a tile into a polyomino of microtiles, keeping the path unchanged at the macro-level. Hence, a similar construction can be used to prove another result, which states that arbitrary-neighborhood tilings can be simulated by von Neumann-neighborhood tilings. As usual, we first prove the result for rectangular neighborhoods, and deduce the general case from that.

\begin{theorem}\label{th:tilings}
  Let $T$ be a tile system in rectangular neighborhood $N_{m,n}$, with set of glues $X$. There exist a simulation $\psi$ and a tile system $T_{\mu}$ in neighborhood $\NVN$, with set of glues $X_{\mu}$, such that any $(T,N_{m,n})$-tiling can be simulated by a $(T_{\mu},\NVN)$-tiling.
\end{theorem}

\begin{proof}
  Only a few modifications have to be done to the general polyomino represented in Fig.~\ref{fig:general}, in the construction of Theorem~\ref{th:rectangular}, to transform the $T_{\mu}$-ribbon into a $(T_{\mu},\NVN)$-tiling. First, the inner layer used to join the input and output microtiles can be removed, allowing thinner spikes.
  
  The main difference is that now, all 4 glues of the microtiles have to match with their neighbors, if present. The glues towards the inside of the polyomino (\ie the glues used for adjacent microtiles from the same polyomino) can be chosen as previously, uniquely for each tile of each polyomino. For the glues used to match with other polyominoes, things are slightly more complex, since they will be used to enforce the correct position of the adjacent polyomino, to avoid misalignment. These glues can be chosen as pairs $(x,y) \in \N^2$. For North and East glues, $(x,y)$ represents the position of the microtile inside the polyomino, counting from the bottom-left part of the polyomino for example (hence, the microtile at the top-left of the polyomino will have glue $(1,s(n+1))$ on the North side). For South and West glues, this pair should be the position of the microtile on the neighboring polyomino, to ensure proper match. Therefore, the West glue of the top-left microtile will be $(s,s(n+1))$, since it has to match with the top-right tile of the polyomino adjacent on the right. Note that $1 \leq x \leq s(m+1)$ and $1 \leq y \leq s(n+1)$ (see Fig.~\ref{fig:general}), therefore the number of these glues is bounded, once $T$ and $N_{m,n}$ are fixed. These glue do not need to encode any information from the original set of glues $X$, since this matching will be enforced by the shape of the polyominoes. Their role is only to align properly the polyominoes.
  
  Then, each $(T,N_{m,n})$-tiling is associated with a set of $(T_{\mu},\NVN)$-tilings by this construction. Conversely, any $(T_{\mu},\NVN)$-tiling $\T'$ can be turned into a poly-tiling $\T''$ using the tile system $T_{\mu}$ in neighborhood $\NVN$, by discarding incomplete and misaligned polyominoes using a function $\chi_{\varphi}$. Then, starting from the poly-tiling $\T''$, one can recover the initial tiles at their respective positions, with the help of the function $\varphi^{-1}$ applied to the unique sets containing each of the complete polyominoes. By construction, this $(T,N)$-tiling $\T$ is the only valid tiling such that $\chi_{\varphi}(\T') \in \psi(\T)$, therefore $\psi$ is a simulation.
\qed\end{proof}

Remark that instead of increasing the number of glues to force a correct alignment of the polyominoes, one could have worked on their shape, adding more bumps and dents as in~\cite{CK09}. We made the choice of a shape as simple as possible, even if it increases the number of glues used.

As explained at the beginning of Sect.~\ref{sec:arb}, any neighborhood can be replaced by an equivalent rectangular neighborhood, hence the following corollary.

\begin{corollary}\label{cor:mainT}
  Let $T$ be a tile system in arbitrary neighborhood $N$, with set of glues $X$. There exist a simulation $\psi$ and a tile system $T_{\mu}$ in neighborhood $\NVN$, with set of glues $X_{\mu}$, such that any $(T,N)$-tiling can be simulated by a $(T_{\mu},\NVN)$-tiling.
\end{corollary}

\subsection{Properties Preserved by the Simulation}

The main interest of this construction is that it guarantees that some properties of the initial tiling are preserved, and still verified in the final poly-tiling. For example, unlike what happens with the simple construction explained in Sect.~\ref{sec:naive}, a partial $(T,N)$-tiling is simulated by a partial $(T_{\mu},\NVN)$-tiling. The converse is also true, as well as some other important properties like periodicity and convexity, as explained in this section.

\paragraph{Partial tilings} Clearly, from our construction, a partial tiling is simulated by partial tilings. For the converse, remark that our construction creates ``holes'' in a poly-tiling, because the polyominoes we build are hollow. Then, provided we fill the polyominoes by new tiles with unique glues, we obtain the following immediate result.

\begin{theorem}\label{th:total}
  Let $T$ be a tile system in arbitrary neighborhood $N$, with set of glues $X$. There exists a simulation $\psi$ such that every $(T,N)$-tiling $\T$ is total if and only if all $(T_{\mu},\NVN)$-tilings in $\psi(\T)$ are total.
\end{theorem}

Note that when $\T$ is total, $\psi(\T)$ is a singleton. Indeed, the polyominoes of $\varphi(t)$ differ only by the position of the bridges. When the tiling is total, all bridges are constrained by the neighboring polyominoes, hence only one element of the set $\varphi(t)$ is possible. The final poly-tiling is therefore uniquely defined.

\begin{remark}
Similarly, remark that a tiling is finite [resp., connected] if and only if all the poly-tilings which simulate it are finite [resp., connected]. As a consequence, the simulation of polyominoes is achieved by polyominoes only.
\end{remark}

\paragraph{Periodic tilings} Periodicity is an essential notion in the classical tiling theory~\cite{B66,R71}, it should be preserved by a meaningful simulation. The following result shows that this is the case with our construction.

\begin{theorem}
  Let $T$ be a tile system in arbitrary neighborhood $N$, with set of glues $X$. There exists a simulation $\psi$ such that every $(T,N)$-tiling $\T$ is periodic if and only if at least one of the $(T_{\mu},\NVN)$-tilings in $\psi(\T)$ is periodic.
\end{theorem}

\begin{proof}
  Clearly, if a $(T_{\mu},\NVN)$-tiling as constructed in Theorem~\ref{th:tilings} is periodic, restricting it to a poly-tiling preserves its periodicity, and then the $(T,N)$-tiling it simulates is also periodic.
  
  Conversely, if a $(T,N)$-tiling is periodic, then one of the poly-tilings which simulate it according to Theorem~\ref{th:tilings} has to be periodic. In the case of a total $(T,N)$-tiling, this fact is obvious. In the case of a partial tiling, it suffices to choose the poly-tiling with the correct bridge constraints at the borders, so that the periodic pattern of the tiling repeats all over.
\qed\end{proof}

Note that if the periods of a tiling $\T$ are $p_v, p_h \in \N_+$, the periods of the periodic poly-tiling in $\psi(\T)$ are $s p_v, s p_h \in \N_+$, where $s$ is the scale of the poly-tiling.

\paragraph{Convex tilings} The different convexity notions (line, column, or both) are important when dealing with polyominoes~\cite{BDNP96,CD99}, which are a particular case of tilings. Our construction does not preserve the convexity of a tiling, but it is possible to modify it in order to preserve line convexity in the case of linear neighborhoods.

\begin{proposition}\label{prop:lconvexity}
  Let $T$ be a tile system in arbitrary linear neighborhood $N$, with set of glues $X$. There exists a simulation $\psi$ such that every $(T,N)$-tiling $\T$ is line convex if and only if all $(T_{\mu},\NVN)$-tilings in $\psi(\T)$ are line convex.
\end{proposition}

\begin{sproof}
  First, polyominoes should be filled as in the proof of Theorem~\ref{th:total}. Then, the construction needs to be modified again, for one reason visible on Fig.~\ref{fig:crossing2}: the bumps, dents and bridges create horizontal holes.
  
%  First of all, the sheaths can be replaced by stacking layers, as shown on Fig.~\ref{fig:sheath}. This requires some technical details to make sure that the polyomino remains connected. Besides, the $\mathbf{B}$ block (Fig.~\ref{fig:general}) would disappear, so the horizontal communications should arrive in $\mathbf{A}$ and the vertical ones in $\mathbf{C}$ (from above).
%    \begin{figure}[!ht]
%    \begin{center}
%    \includegraphics{figs/sheath}
%    \caption{Replacing sheaths by layers. The different polyominoes are represented by different gray levels, they stack instead of interlocking.}
%    \label{fig:sheath}
%    \end{center}
%  \end{figure}

  The idea is to replace these ``holes'' by two sets of stairs, using again the position of the stairs to ensure the glue matching, as represented in Fig.~\ref{fig:dent}. In this figure, the integer $g$ is the position of the dent, as induced by some glue. In order for the top and bottom polyominoes to match, the bump has to be located at position $g+k+1$, where $k$ is the number of layers crossed. In the new construction, the glue matching is enforced by the fact that in the bottom polyomino, if $g$ is too big then there is an overlap on the left stairs, and if $g$ is too small then there is an overlap on the right stairs.
    \begin{figure}[!ht]
    \begin{center}
    \includegraphics{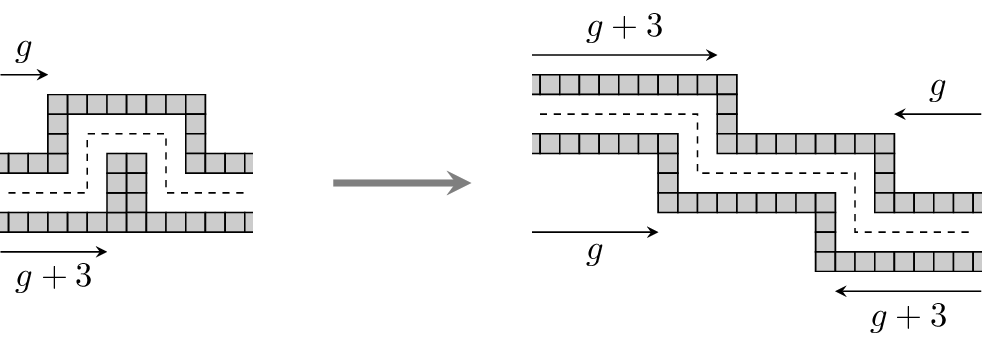}
    \caption{Replacing bumps and dents by two stairs. Here, there are $k=2$ layers to be crossed.}
    \label{fig:dent}
    \end{center}
  \end{figure}

  The problem with this technical step is that crossing $k$ layers implies that the spike goes down by $2(k+1)$. If $\gamma$ is the number of communication places in a segment of length $l$, and $\delta$ is the maximum number of layers which can be crossed, then each segment of length $l$ in the spike would go down by $2\gamma (\delta+1)$ instead of~$1$. This imposes new constraints on $h$ and $s$, but the discretization system is still solvable, at the additional cost of increasing $l$ by a factor $2\gamma (\delta+1)$.
  
  Then, obviously, all polyominoes from the sets $\varphi(t)$ are line convex. Moreover, all polyominoes in any set $\varphi(t)$ have the same height~$s$ and are assumed to have the same reference microtile, so that $\psi$ transforms a line of tiles into a block of microtiles of height~$s$. This block is formed by the horizontal juxtaposition of polyominoes, which, by construction, do not leave holes between them (otherwise, there is a glue mismatch somewhere, leading to some overlap somewhere else, see Fig.~\ref{fig:dent}). Thus, this block is made of $s$ lines of microtileswithout holes, and is therefore line convex. Since a line convex tiling is a vertical superposition of lines, its image by $\psi$ is a vertical superposition of line convex blocks, hence another line convex tiling. Finally, the converse is trivial, if a $(T_{\mu},\NVN)$-tiling is line convex then it simulates a line convex $(T,N)$-tiling.
\qed\end{sproof}

A similar result can be stated for column convex tilings, provided the initial tiling uses tiles in a \emph{vertical} neighborhood $N$: if $(i,j) \in N$ then $i=0$. The proof follows the same sketch as for Proposition~\ref{prop:lconvexity}, with everything rotated by~90\textdegree.

\begin{proposition}
  Let $T$ be a tile system in arbitrary vertical neighborhood $N$, with set of glues $X$. There exists a simulation $\psi$ such that every $(T,N)$-tiling $\T$ is column convex if and only if all $(T_{\mu},\NVN)$-tilings in $\psi(\T)$ are column convex.
\end{proposition}

%%%%%%%%%%%%%%%%%%%%%%%%%%%% CONCLUSIONS %%%%%%%%%%%%%%%%%%%%%%%%%%%%

\section{Conclusions and Perspectives}\label{sec:concl}

In this paper, we proved that any zipper formed with tiles defined in an arbitrarily complex neighborhood can be simulated by a ribbon obtained by the catenation of simple polyominoes. Each of the new microtiles used for the simulation has, when placed on a ribbon, only two neighbors: its predecessor and its successor on the path. A similar construction can be achieved for tilings, in order to replace arbitrary neighborhoods by the simpler von Neumann neighborhood, and at the same time preserve some properties of the tiling.

A few research directions can extend this work. For example, the simulation relies on the fact that zippers and ribbons are not self-crossing. Although this is the standard way to define them, they could be generalized to other, self-crossing, notions. It would be interesting to see if similar results could be obtained in this case, using new constructions. Another interesting topic would be to study how this construction behaves in three-dimensions, since crossings might be avoided with spikes turning around each other. However, new problems occur and would have to be solved by original techniques.

\smallskip

Another possible improvement is justified by the observation that our constructions are not yet adapted to practical dynamical implementations, such as DNA self-assembly. Theoretically, our constructions could be used for practical purposes: for example, as mentioned in~\cite{R01}, the construction of Theorem~\ref{th:VN} can be used to simulate a Turing machine at temperature~$1$ (\ie DNA tiles attach to the growing assembly whenever one glue matches). But this produces a lot of ``garbage'' consisting of many blocked polyominoes, for a limited number of correct assemblies. Indeed, there are many dynamical scenarios wherein, for example, the self-assembly of a ``bad'' polyomino is started, that blocks any further aggregation. This issue could be solved by adding the possibility to recover from a wrong start, for example by allowing tiles to ``unstick'', as suggested in~\cite{W98,A00}. Further modifications would then be necessary to guarantee that, also in this setting, our constructions can self-assemble fully and correctly in finite time.

Also remark that in our constructions, we emphasized the regularity of the polyominoes, to provide easier general constructions and aggregation. This led however to a blow-up in the number of new microtiles necessary to simulate a zipper, potentially leading to more experimental issues. The next step would be to optimize the construction according to some of the following criteria: number of polyominoes associated with a tile, number of microtiles appearing in a polyomino, number of glues used to build the polyominoes, \etc

An alternative solution to these two problems would be the use of staged self-assembly~\cite{DDFIRSS08} for experiments. This formalism would allow to add an initial stage to construct the polyominoes in separate bins, before mixing them in the final solution, preventing the formation of ``bad'' assemblies. Besides, staged self-assembly would dramatically decrease the number of required glues (hence the number of microtiles), by adding even more initial stages to produce the polyominoes. This would offer more control on the order in which microtiles attach, thus removing the necessity for unique glues. However, this would be achieved at the cost of increased control by an external operator during the self-assembly process. This trade-off between external control and tile complexity is often encountered in DNA self-assembly (see, \eg~\cite{KS06,DDFIRSS08}), and is currently being investigated by the authors.

\smallskip

The authors would like to thank Shinnosuke Seki for his useful suggestion simplifying the construction in Fig.~\ref{fig:general}, and Leonard Adleman, Jarkko Kari, and Erik Winfree for discussions.

\bibliographystyle{abbrv}
\bibliography{arb}

\end{document}